\documentclass[a4paper]{article}

\usepackage{a4wide}
\usepackage{bbold}
\usepackage[utf8]{inputenc} 
\usepackage[T1]{fontenc}    
\usepackage{hyperref}       
\usepackage{amsfonts}       
\usepackage{graphicx}
\usepackage{amsmath}
\usepackage{amsthm}
\usepackage{float}
\usepackage{subcaption}
\usepackage{mathtools}
\usepackage[round]{natbib}
\graphicspath{{images/}{./}}

\newcommand{\E}{\mathbb{E}}
\renewcommand{\P}{\mathbb{P}}
\newcommand{\R}{\mathbb{R}}
\newcommand{\cA}{\mathcal{A}}
\newcommand{\cD}{\mathcal{D}}
\newcommand{\cF}{\mathcal{F}}
\newcommand{\cP}{\mathcal{P}}
\newcommand{\cR}{\mathcal{R}}
\newcommand{\cS}{\mathcal{S}}
\newcommand{\cT}{\mathcal{T}}
\newcommand{\ind}[1]{\mathbb{1}_{\left\{#1\right\}}}
\newcommand{\esssup}{\mathop{\mathrm{esssup}}}

\newtheorem{theorem}{Theorem}[section]
\newtheorem{lemma}[theorem]{Lemma}
\newtheorem{proposition}[theorem]{Proposition}

\DeclarePairedDelimiter{\abs}{\lvert}{\rvert}
\title{Pricing Bermudan options using Regression Trees / Random Forests}

\author{
Zineb El Filali Ech-chafiq \thanks{Univ. Grenoble Alpes, CNRS,
Grenoble INP, LJK,
38000 Grenoble, France.
Quantitative analyst at Natixis, Paris.
zinebelfilaliechchafique@gmail.com}
\and
Pierre Henry Labordère \thanks{Head of Quantitative Research cross asset,
Natixis, 47 quai d'austerlitz
75013, Paris.
CMAP, Ecole Polytechnique}
\and
Jérôme Lelong \thanks{Univ. Grenoble Alpes, CNRS,
Grenoble INP, LJK,
38000 Grenoble, France.
jerome.lelong@univ-grenoble-alpes.fr}
}
 
\begin{document}
\maketitle

\begin{abstract}
 The value of an American option is the maximized value of the discounted cash flows from the option. At each time step, one needs to compare the immediate exercise value with the continuation value and decide to exercise as soon as the exercise value is strictly greater than the continuation value. We can formulate this problem as a dynamic programming equation, where the main difficulty comes from the computation of the conditional expectations representing the continuation values at each time step. In \citep{LongstaffSchwarz}, these conditional expectations were estimated using regressions on a finite-dimensional vector space (typically a polynomial basis). In this paper, we follow the same algorithm; only the conditional expectations are estimated using Regression trees or Random forests. We discuss the convergence of the LS algorithm when the standard least squares regression is replaced by regression trees. Finally, we expose some numerical results with regression trees and random forests. The random forest algorithm gives excellent results in high dimensions. 
\end{abstract}

{\bf keywords}: Regression trees, Random forests, Bermudan options, Optimal stopping

\section{Introduction}

Bermudan options are very widespread in financial markets. Their valuation adds a challenge of optimal stopping determination in comparison to European options. Bermudan options offer the investor the possibility to exercise his option at any date of his choice among a certain number of dates prior to the option expiry, called exercise dates. Naturally, the option holder will have to find the most optimal date to exercise. To do so, at each exercise date, he will compare the payoff of the immediate exercise to the expected value of continuation of the option and decide to exercise only if the immediate exercise value is the highest. We can formulate this problem as a dynamic programming equation, where the main difficulty comes from the computation of the conditional expectation representing the expected continuation value of the option. Many papers have discussed this issue, starting with regression-based algorithms; see for example \citep{Tsitsiklis} and \citep{carriere}. Also, in this category falls the most commonly used method for pricing Bermudan options which is the Least Squares Method (LSM) presented by Longstaff and Schwarz in  \citep{LongstaffSchwarz} where the conditional expectation is estimated by a least squares regression of the post realized payoffs from continuation on some basis functions of the state variables (usually polynomial functions). Alternatively, an approximately optimal but interpretable policy can be obtained as the solution to a sample average approximation problem, see~\cite{ciocan2022interpretable}. Another class of algorithms focuses on quantization approaches, see for example \citep{Bally}. The algorithm consists in computing the conditional expectations by projecting the diffusion on some optimal grid. We also have a class of duality based methods that give an upper bound on the option value for a given exercise policy by adding a quantity that penalizes the incorrect exercise decisions made by the sub-optimal policy, see for example \citep{rogers}, \citep{Broadie} and \citep{chaos}. The last class of algorithms is based on machine learning techniques. Deep learning has been widely used to solve stochastic control problems \cite{DBLP:journals/corr/HanE16} or their PDE formulation \cite{han2018solving}. Neural networks can also be used to estimate the continuation values as in \citep{kohler, han2018solving, Lelong} or the boundary \citep{reppen2022neural}. The continuation value can also be approximated using Gaussian processes as in \citep{Ludkovski}. Our solution falls in this last category of algorithms. We examine Bermudan options' prices when the continuation values' estimation is done using regression trees or random forests.\par

Let $X, Y$ be two random variables with vales in $[0,1]^d$ and $\mathbb{R}$ respectively. A regression tree approximates the conditional expectation $\mathbb{E}\left[Y|X \right]$ with a piecewise constant function. The tree is built recursively, generating a sequence of partitions of $[0,1]^d$ that are finer and finer. The approximation value on each set in the partition can be seen as a terminal leaf of the tree. This algorithm is very simple and efficient. However, it can easily over-fit the data, which results in high generalization errors. To solve this issue, we use ensemble methods to aggregate multiple trees, which means that we create multiple trees and then combine them to produce improved results. We suggest using random forests (see \citep{Breiman}). This method consists in averaging a combination of trees where each tree depends on a random vector sampled independently and identically for each tree in the forest. This vector will allow to differentiate the trees in the random forest and can be chosen in different ways. For example, one can draw for each tree a sub-sample of training from the global training data without replacement (this method is called bagging and is thoroughly studied in \citep{bagging}). A second method is random split selection, where at each node, the split is selected at random from among the $K$ best splits, see \citep{Dietterich}. Other methods for aggregating regression trees into random forests can be found in the literature, see for example \citep{Breiman} or \citep{ho}.\\
The structure of the paper will be as follows. First, we present the regression trees algorithm and the algorithm of least squares using regression trees. Then, we proceed to present  some convergence results for regression trees and study the convergence of the LS algorithm when regression trees are used to estimate the continuation values. We prove two results: the convergence of the regression solution for infinite sample size as the depth of the regression tree goes to infinity and the convergence of the empirical price as the number of samples go to infinity for a fixed depth. Then, we briefly talk about Random Forests before we finally study some numerical examples.

\section{Regression trees}
\label{section2}
Let $X$ be a r.v with values in $[0,1]^d$ and $Y$ a squared integrable real-valued r.v. We want to approximate the conditional expectation $\mathbb{E}[Y|X]$. Throughout this paper, we will consider for computational convenience that $X$ has a density $f_X$ in $[0,1]^d$ w.r.t the Lebesgue measure. We assume given a training sample $D_M=\{(X_1, Y_1), \ldots, (X_M, Y_M) \in [0,1]^d \times \mathbb{R} \}$ where the $(X_i, Y_i)$'s are i.i.d random variables following the law of $(X, Y)$. An approximation using a regression tree consists in writing the conditional expectation as a piecewise constant function of $X$. Each domain where the function is constant can be seen as a terminal leaf of a tree. Let us start with the one-dimensional case $(d=1)$. Consider the function
\begin{equation*}
    m_M: y_l \in \R, y_r \in \R, x \in [0, 1] \longmapsto \frac{1}{M} \sum_{i=1}^M (Y_i - y_l \ind{X_i \le x} - y_r \ind{X_i > x})^2.
\end{equation*}
With probability $q >0$, choose $x$ as the midpoint $x=\frac{1}{2}$ and solve the minimization problem $\inf_{y_l, y_r} m_M(y_l, y_r, \frac{1}{2})$. With probability $0 < 1 - q < 1$, solve the minimization problem $\inf_{y_l, y_r, x} m_M(y_l, y_r, x)$. For a fixed $x^*$, the optimal values of $y_l$ and $y_r$ are given by 
\begin{equation}
\label{eq:yLyR}
y_r^*= \frac{\sum_{i=1}^M Y_i \mathbb{1}_{\{X_i > x^*\}}}{\sum_{i=1}^M \mathbb{1}_{\{X_i > x^*\}}}; \quad
y_l^* = \frac{\sum_{i=1}^M Y_i \mathbb{1}_{\{X_i \leq x^*\}}}{\sum_{i=1}^M \mathbb{1}_{\{X_i \leq x^*\}}}.
\end{equation}
The problem $\inf_{y_l, y_r, x} m_M(y_l, y_r, x)$ may admit more than one minimizer. To make the tree unique, we choose the minimizer $(y_l^*, y_r^*, x^*)$ with the minimal third component, ie any solution $(y_l, y_r, x)$ to $\inf_{y_l, y_r, x} m_M(y_l, y_r, x)$ satisfies $x > x^*$.

Once the threshold $x^*$ is determined, we split the samples into two groups following the sign of $X_i - x^*$ and repeat the process for each group. We stop the process if introducing a new leaf does not improve the MSE (ie when $x^*$ is on the boundary of the interval on which we solve the minimization problem) or when enough iterations have been made. In the end, we have a tree that approximates the conditional expectation with a piecewise constant function. The regression trees are an algorithmic tool to find an adapted partition and the corresponding weights of this piecewise constant function. \\
In the multi-dimensional case, we choose the direction (the index along which the optimization is performed) uniformly for each new split. Then, the process is iterated as in the one-dimensional case. 

We denote the resulting tree by $\hat{\mathcal{T}}_p^M(X)$ where $p$ represents the depth of the tree, i.e., the number of iterations done in the process described above. A tree of depth $p$ has $2^p$ leaves.\\
When the size of the training data is infinite, we use the same procedure as above but with expectations rather that empirical means. Consider the function
\begin{equation*}
    m: y_l \in \R, y_r \in \R, x \in [0, 1] \longmapsto \E\left[ (Y - y_l \ind{X \le x} - y_r \ind{X > x})^2\right].
\end{equation*}
Equation~\eqref{eq:yLyR} becomes
\begin{equation*}
    y_r^* = \mathbb{E}\left[ Y | X >x^*\right]; \quad
    y_l^* = \mathbb{E}\left[ Y | X \leq x^*\right].
\end{equation*}
Again, when we optimize the value of $x^*$, we choose the minimizer $(y_l^*, y_r^*, x^*)$ with the minimal third component.  We denote the tree of depth $p$ obtained with an infinite data set by $\mathcal{T}_p(X)$.


\section{LS algorithm with regression trees}
\subsection{Notation}

For $p \in \mathbb{N}$, let $\cA = \left(\prod_{j=1}^d \left[a_p^{i-1}(j), a_p^i(j)\right)\right)_{1 \leq i \leq 2^p}$ be a partition of $[0,1]^d$ with $2^p$ elements. We write 
\begin{equation*}
    \left[a_p^{i-1} , a_p^i\right[ := \prod_{j=1}^d \left[a_p^{i-1}(j) , a_p^i(j)\right[
\end{equation*}
For $(\alpha_p^i)_{1 \le i \le 2^p} \in \R^{2^p}$, we define $\cP_p$ as the piecewise constant function on the partition $\cA$ with values $\alpha_p^i$. For $x \in [0,1]^d$,
\begin{equation*}
    \cP_p(x, (a_p^i)_{0 \le i \le 2^p}, (\alpha_p^i)_{1 \le i \le 2^p}) = 
    \sum_{i=1}^{2^p} \alpha_p^i \mathbb{1}_{\{x \in [a^{i-1}_p, a^{i}_p)\}}
\end{equation*}

If $\cA$ denotes the partition obtained in the regression tree $\cT_p$, and we choose 
\begin{equation*}
    \alpha_p^i = \mathbb{E}\left[Y|X\in \left[a_p^{i-1}, a^i_p\right) \right],
\end{equation*}
then the regression tree $\mathcal{T}_p(X)$ can be written as follows
\begin{equation}
    \label{eq:Tp_P}
    \mathcal{T}_p(X) = \cP_p(X, (a_p^i)_{0 \le i \le 2^p}, (\alpha_p^i)_{1 \le i \le 2^p})
\end{equation}

Similarly, if $\cA^M = \left(  \left[a_p^{i-1, M} , a_p^{i,M}\right)  \right)_{1 \le i \le 2^p}$ denotes the partition obtained in the regression tree $\hat \cT_p^M$ and \begin{equation*}
    \alpha_p^{i,M} = \frac{\sum_{m=1}^M Y_m \ind{X_m \in \left[a_p^{i-1}, a^i_p\right)}}{\sum_{m=1}^M \ind{X_m \in \left[a_p^{i-1,M}, a^{i,M}_p\right)}},
\end{equation*}
then the regression tree $\hat{\mathcal{T}}^M_p(X)$ can be written as follows
\begin{equation}
    \label{eq:hatTp_P}
    \hat{\mathcal{T}}^M_p(X) = \cP_p(X, (a_p^{i,M})_{0 \le i \le 2^p}, (\alpha_p^{i,M})_{1 \le i \le 2^p})
\end{equation}

\subsection{Description of the algorithm}
Let $T$ be a fixed maturity, and consider the filtered probability space $(\Omega, \mathcal{F}, (\mathcal{F}_t)_{0 \leq t \leq T}, \mathbb{P})$ where $\mathbb{P}$ is the risk neutral measure. Consider a Bermudan option that can be exercised at dates $0=t_0 < t_1 < t_2 < \cdots < t_N = T$. When exercised at time $t_j$, the option's discounted payoff is given by $Z_{t_j} = h_j(X_{t_j})$ with $(X_{t_j})_j$ being an adapted Markov process taking values in $[0,1]^d$ such that for every $j$, $X_{t_j}$ has a density $f^j_X$ on $[0,1]^d$ uniformly bounded from below on any compact set of $]0,1[^d$ by a strictly positive number. The discounted value $(U_{t_j})_{0\leq j\leq N}$ of this option is given by 
\begin{equation}
    U_{t_j} = \esssup_{\tau \in \cS_{t_j, T}} \mathbb{E}\left[ Z_{\tau} | \mathcal{F}_{t_j}\right],
\end{equation}
where $\cS_{t_j, T}$ is the set of all stopping times taking values in $\{t_j,\dots,T\}$.  Using the Snell envelope theory, we know that $U$ solves the dynamic programming equation

\begin{equation}
\left\{
    \begin{array}{ll}
        U_{t_N} & = Z_{t_N}\\
        U_{t_j} & = \max\left(Z_{t_j}, \mathbb{E}\left[ U_{t_{j+1}} | \mathcal{F}_{t_j}\right]\right) \textnormal{ for } 1\leq j \leq N - 1.
    \end{array}
\right.
\end{equation}

This equation can be rewritten in terms of optimal policy as follows

\begin{equation}
\left\{
    \begin{array}{ll}
        \tau_N & = t_N = T\\
        \tau_j & = t_j \mathbb{1}_{\left\{ Z_{t_j} \geq \mathbb{E}\left[ Z_{\tau_{j+1}}|\mathcal{F}_{t_j}\right]\right\}} + \tau_{j + 1} \mathbb{1}_{\left\{ Z_{t_j} < \mathbb{E}\left[Z_{\tau_{j+1}}| \mathcal{F}_{t_j}\right]\right\}} \textnormal{ for } 1\leq j \leq N - 1
    \end{array}
\right.
\end{equation}

where $\tau_j$ is the smallest optimal stopping time after $t_j$. As we are in a Markovian setting, we can write $\mathbb{E}\left[Z{\tau_{j+1}} | \mathcal{F}_{t_j}\right] = \mathbb{E}\left[Z_{\tau_{j+1}} | X_{t_j} \right]$. The main difficulty in solving this equation comes from the computation of the continuation value $\mathbb{E}\left[ Z_{\tau_{j+1}} | X_{t_j} \right]$ . In the Least Squares approach presented by \citep{LongstaffSchwarz}, this conditional expectation is estimated by a linear regression on a countable set of basis functions of $X_{t_j}$. In our approach, we suggest to estimate  it using a regression tree of depth $p$. Let $\mathcal{T}_p^j(X_{t_j})$ be the regression tree approximating $\mathbb{E}\left[ Z_{\tau_{j+1}} | X_{t_j} \right]$ with an infinite data set. The algorithm solves the following policy

\begin{equation}
\label{eq:dpp-taup}
\left\{
    \begin{array}{ll}
       \tau_N^p & = t_N = T\\
        \tau_j^p & = t_j \mathbb{1}_{\left\{ Z_{t_j} \geq \mathcal{T}_p^j(X_{t_j})\right\}} + \tau_{j + 1} \mathbb{1}_{\left\{ Z_{t_j} < \mathcal{T}_p^j(X_{t_j})\right\}} \textnormal{  for } 1\leq j \leq N - 1.
    \end{array}
\right.
\end{equation}

We sample $M$ paths of the model $X_{t_0}^{(m)}, \ldots, X_{t_N}^{(m)}$ along with the corresponding payoff paths $Z_{t_0}^{(m)},\ldots, Z_{t_N}^{(m)}$, $m=1,\ldots, M$. For each date $j=1,\ldots, N-1$ we approximate the conditional expectations $\mathbb{E}[Z_{\tau_{j+1}}|X_{t_j}]$ on the path $m$ using the regression tree $\hat{\cT}^{j,M}_p(X^{(m)}_{t_j})$ built with the samples $(X_{t_j}^{(m)}, Z^{m}_{\hat{\tau}_{j+1}^{p,(m)}})_{1 \le m \le M}$ where
\begin{equation}
        \label{equation6}
        \left\{
            \begin{array}{ll}
                \hat{\tau}_N^{p, (m)} & = t_N = T\\
                \hat{\tau}_j^{p,(m)} & = t_j \mathbb{1}_{\left\{ Z_{t_j}^{(m)} \geq \hat{\mathcal{T}}_p^{j,M}(X_{t_j}^{(m)})\right\}} + \hat{\tau}_{j+1}^{(m)} \mathbb{1}_{\left\{ Z_{t_j}^{(m)} < \hat{\mathcal{T}}_p^{j, M}(X_{t_j}^{(m)})\right\}} \textnormal{ for } 1\leq j \leq N - 1.
            \end{array}
        \right.
    \end{equation}
Finally, the time-0 price of the option is approximated by
\begin{equation}
    \label{equation7}
    U_0^{p,M} = \max\left(Z_0, \frac{1}{M}\sum\limits_{m=1}^M Z_{\hat{\tau}_1^{p,(m)}}^{(m)}\right).
\end{equation}

\section{Convergence of the algorithm}
\subsection{Notation}

For each time step $1 \le j \le N-1$, let $\theta_j^p = ((a^p_{i,j})_{0 \le i \le 2^p}, (\alpha^p_{i,j})_{1 \le i \le 2^p})$ be the parameters of $\cT_p^j$ and $\hat{\theta}_j^{p,M} = ((a^{p,M}_{i,j})_{0 \le i \le 2^p}, (\alpha^{p,M}_{i,j})_{1 \le i \le 2^p})$ be the parameters of $\cT_{p,M}^j$. Note that with this notation and Equations~\eqref{eq:Tp_P} and~\eqref{eq:hatTp_P}, we have
\begin{align}
    \label{eq:link-P}
    \mathcal{T}_p^j(X_{t_j}) = \cP_p(X_{t_j}, \theta_j^p) \\
    \hat{\mathcal{T}}_p^{j,M}(X_{t_j}^{(m)}) = \cP_p(X_{t_j}^{(m)}, \hat{\theta}_j^{p,M})
\end{align}
where for te sake of conciseness we have shrunk the partition and weight parameters of the tree into a single multi-dimensional parameter.

We introduce the vector $\vartheta$ of the coefficients of the successive expansions $\vartheta^p = (\theta_1^p,\ldots, \theta_{N-1}^p)$ and $\hat{\vartheta}^{p,M} = (\hat{\theta}_1^{p,M},\ldots, \hat{\theta}_{N-1}^{p,M})$.\\
Let $t^p = (t_1^p,\ldots, t_{N-1}^p) \in \left(([0,1]^d)^{2^p + 1}\right)^{N-1}$ be a deterministic parameter, $z = (z_1,\ldots, z_N) \in \mathbb{R}^N$ and $x=(x_1,\ldots, x_N)\in ([0,1]^d)^N$ be deterministic vectors. Following the notation of \citep{clement}, we define the function $F=(F_1,\ldots,F_N)$ taking values in $\R^N$ by 
\begin{equation*}
    \left\{
            \begin{array}{ll}
                F_N(t^p,z,x) &= z_N\\
                F_j(t^p,z,x) &= z_j\mathbb{1}_{\{z_j \geq \mathcal{P}_p(x_j, t_j^p)\}} + F_{j+1}(t^p,z,x)\mathbb{1}_{\{z_j < \mathcal{P}_p(x_j,t_j^p)\}}, \textnormal{ for } 1 \leq j \leq N-1.
            \end{array}
        \right.
\end{equation*}
$F_j(t^p,z,x)$ only depends on $t_j^p, \ldots, t_{N-1}^p$ and not on the first $j - 1$ components of $t^p$. Using~\eqref{eq:link-P}, we have
\begin{align*}
    F_j(\vartheta^p,Z,X) &= Z_{\tau_j^p},\\
    F_j(\hat{\vartheta}^{p,M},Z^{(m)},X^{(m)}) &= Z_{\hat{\tau}_j^{p,(m)}}^{(m)}.
\end{align*}
Moreover, we clearly have that for all $t^p \in (([0,1]^d)^{2^p})^{N-1}$:
\begin{equation}
    \label{equation8}
    \abs*{F_j(t^p,Z,X)} \leq \max\limits_{k\geq j}\abs*{Z_{t_k}}.
\end{equation}
\subsection{Convergence of the conditional expectations}
\subsubsection{Preliminary results}

We define
\begin{equation*}
    J_p\left(a_p^0, \ldots, a_p^p\right) := \mathbb{E} \left[ \abs*{\sum_{i=1}^p \alpha^i_p \mathbb{1}_{\{X \in \left[a^{i-1}_p - a^i_p\right)\}} - \mathbb{E}\left[Y|X\right]}^2\right].
\end{equation*}
with $\alpha_p^i = \E[Y | X \in [a^{i-1}_p - a^i_p)]$. We will use the following standard lemma to prove the convergence of random tree estimators.  \begin{lemma}
    \label{lemma4.1}
    Let $X$ be a r.v with a density $f_X$ w.r.t to the Lebesgue measure on $[0,1]^d$ and $Y$ be a real valued square integrable random variable. Let $\left(\left( \left[a^{i-1}_p, a^i_p\right)\right)_{1 \leq i \leq p}\right)_{p \in \mathbb{N}}$ be a sequence of partitions of $[0,1]^d$ such that  $\lim\limits_{p \to \infty} \max\limits_{1\leq i\leq p}\max\limits_{1\leq j\leq d} \abs*{a^i_p(j) - a^{i-1}_p(j)} = 0 $.
    Then, 

    \begin{equation*}
        \lim_{p \to \infty} J_p\left(a^0_p,\ldots, a^p_p\right)= 0.
    \end{equation*}
\end{lemma}

\begin{theorem}
\label{theorem4.2}
\begin{equation*}
     \lim_{p\to\infty}\mathbb{E}\left[\abs*{\mathcal{T}_p(X) - \mathbb{E}[Y|X]}^2\right] = 0.
\end{equation*}
\end{theorem}

\begin{proof}
With probability $\frac{1}{d}$, the index $j$ is chosen for optimisation. In the other $d-1$ cases, we do not even cut along that direction, in which case the interval length is at most equal to the length of the largest interval at step $p-1$. When the index $j$ is chosen: with probability $q$, the length of the interval is cut in two, and with probability $1 - q$, it is cut to optimize the MSE. In that case, the interval length is at most equal to the length of the largest interval at step $p-1$.

For all $1 \leq j \leq d$
\begin{align*}
    &\mathbb{E}\left[\max_{1 \leq i \leq 2^p}\abs*{a^i_p(j) - a^{i-1}_p(j)}\right]\\ &\leq \frac{1}{d}\left[ q \frac{1}{2}\mathbb{E}\left[\max_{1 \leq i \leq 2^{p - 1}}\abs*{a^i_{p - 1}(j) - a^{i-1}_{p-1}(j)}\right] + (1 - q) \mathbb{E}\left[\max_{1 \leq i \leq 2^{p - 1}}\abs*{a^i_{p - 1}(j) - a^{i-1}_{p-1}(j)}\right]\right] \\
    &+ \frac{d-1}{d}\mathbb{E}\left[\max_{1 \leq i \leq 2^{p - 1}}\abs*{a^i_{p - 1}(j) - a^{i-1}_{p - 1}(j)}\right]\\
    &\leq (1 - \frac{q}{2d}) \mathbb{E}\left[\max_{1 \leq i \leq 2^{p - 1}}\abs*{a^i_{p-1}(j) - a^{i-1}_{p-1}(j)}\right]\\
    &\leq (1 - \frac{q}{2d})^{2^p}
\end{align*}
Since $\sum\limits_{p=0}^{\infty}(1 - \frac{q}{2d})^{2^p} < \infty$, so is $\sum\limits_{p=0}^{\infty} \mathbb{E}\left[\max\limits_{1\leq i\leq 2^p}\abs*{a^i_p(j) - a^{i-1}_p(j)}\right]$.
As $\max\limits_{1\leq i\leq 2^p}\abs*{a^i_p(j) - a^{i-1}_p(j)}$ is non negative for all $p$, using Tonelli's theorem we conclude that $\mathbb{E}\left[\sum\limits_{p=0}^{\infty}\max\limits_{1\leq i\leq 2^p}\abs*{a^i_p(j) - a^{i-1}_p(j)} \right] < \infty$. As a result, the series $\sum\limits_{p=0}^{\infty}\max\limits_{1\leq i\leq 2^p}\abs*{a^i_p(j) - a^{i-1}_p(j)}$  converges a.s. Then, $\lim\limits_{p\to \infty} \max\limits_{1\leq i\leq 2^p} \abs*{a^i_p(j) - a^{i-1}_p(j)} = 0$ a.s for all $j$ and $\lim\limits_{p \to \infty} \max\limits_{1 \leq i \leq 2^p} \max\limits_{1 \leq j \leq d} \abs*{a^i_p(j) - a^{i-1}_p(j)} = 0$.\\\\
Let $\mathcal{G}$ be the $\sigma$-field generated by the splitting strategy (direction choice and threshold strategy).  Conditioning by $\mathcal{G}$ allows us to consider the partition $\left([a^{i-1}_p - a^i_p)\right)_{1 \leq i\leq 2^p}$ deterministic and we can apply Lemma \ref{lemma4.1} to prove

\begin{equation*}
    \lim_{p \to \infty}  \mathbb{E}\left[\abs*{\mathcal{T}_p(X) - \mathbb{E}\left[Y|X\right]}^2| \mathcal{G}\right]= 0 \text{ a.s}.
\end{equation*}

Note that
\begin{align*}
    \mathbb{E}\left[ \abs*{\mathcal{T}_p(X)}^2  | \mathcal{G}\right] &\leq \mathbb{E}\left[ \sum_{i=1}^{2^p} \mathbb{E}\left[Y^2|X\in[a^{i-1}_p, a^i_p)\right] \mathbb{1}_{\{X \in [a^{i-1}_p, a^i_p)\}}| \mathcal{G}\right] \\
    &\leq  \sum_{i=1}^{2^p} \mathbb{E}\left[ \mathbb{E}\left[Y^2 \mathbb{1}_{\{X \in [a^{i-1}_p, a^i_p)\}}|X\in[a^{i-1}_p, a^i_p)\right] | \mathcal{G}\right] \\
    &\leq  \sum_{i=1}^{2^p} \mathbb{E}\left[Y^2 \mathbb{1}_{\{X \in [a^{i-1}_p, a^i_p)\}}| \mathcal{G}\right] \\
    &\leq \mathbb{E}\left[Y^2|\mathcal{G}\right]
\end{align*}

Then,

\begin{align*}
     \mathbb{E}\left[\abs*{\mathcal{T}_p(X) - \mathbb{E}\left[Y|X\right]}^2| \mathcal{G}\right] &\leq 2  \mathbb{E}\left[ \abs*{\mathcal{T}_p(X)}^2  | \mathcal{G}\right] + 2 \mathbb{E}\left[ \mathbb{E}\left[ Y|X\right]^2\right]\\
     &\leq 2 \left( \mathbb{E}\left[Y^2|\mathcal{G}\right] + \mathbb{E}\left[ \mathbb{E}\left[ Y|X\right]^2\right]\right)
\end{align*}

Using Lebesgue's bounded convergence theorem, 

\begin{equation*}
    \lim_{p\to\infty}\mathbb{E}\left[\abs*{\mathcal{T}_p(X) - \mathbb{E}[Y|X]}^2\right] = \lim_{p \to \infty} \mathbb{E}\left[ \mathbb{E}\left[\abs*{\mathcal{T}_p(X) - \mathbb{E}[Y|X]}^2| \mathcal{G}\right]\right] = 0.
    \qedhere
\end{equation*}
\end{proof}

Once the partition of the tree is computed, the values associated to each element of the partition are uniquely determined by a global optimization problem
\begin{proposition}
    \label{prop:optim-tree-weights}
    Let $\cA = \left(\prod_{j=1}^d \left[a_p^{i-1}(j), a_p^i(j)\right)\right)_{1 \leq i \leq 2^p}$ denote the partition associated to $\cT_p$. Then,
    \begin{equation*}
        \E\left[ \left( \cT_p(X) - \E[Y|X] \right)^2 \right] = \inf_{(\alpha_p^i)_{1 \le i \le 2^p}} \E\left[ \left( \cP_p(X, (a_p^i)_{0 \le i \le 2^p}, (\alpha_p^i)_{1 \le i \le 2^p} ) - \E[Y|X] \right)^2 \right]
    \end{equation*}
    and moreover the right hand side admits a unique minimizer given by
    \begin{equation*}
        \alpha_p^i = \mathbb{E}\left[Y|X\in \left[a_p^{i-1}, a^i_p\right) \right].
    \end{equation*}
\end{proposition}
\begin{proof}
    Since the function $\cP_p$ is linear w.r.t the parameter $\alpha$, the function \[\alpha \in \R^{2^p+1} \longmapsto \E\left[ \left( \cP_p(X, (a_p^i)_{0 \le i \le 2^p}, (\alpha_p^i)_{1 \le i \le 2^p} ) - \E[Y|X] \right)^2 \right]\] is strongly convex. Hence, it admits a unique minimizer defined by the first order optimality condition stating that for all $0 \le i \le 2^p+1$,
    \begin{align*}
        &\E\left[ \left( \cP_p(X, (a_p^i)_{0 \le i \le 2^p}, (\alpha_p^i)_{1 \le i \le 2^p} ) - \E[Y|X] \right) \ind{X \in \left[a_p^{i-1}, a^i_p\right) } \right] = 0 \\
        &\E\left[ \left( \alpha_p^i - \E[Y|X] \right) \ind{X \in \left[a_p^{i-1}, a^i_p\right) } \right] = 0\\
        &\alpha_p^i = \mathbb{E}\left[Y|X\in \left[a_p^{i-1}, a^i_p\right) \right].\qedhere
    \end{align*}
\end{proof}
A similar result holds for $\hat{\cT}_p^M$ by replacing the conditional expectations by empirical conditional expectations.

\subsubsection{Approximation of the conditional expectations with regression trees}

The following result proves that the approximation of the Bermudan price produced by our algorithm converges to the true Bermudan price.
\begin{proposition}
\label{prop:pricecv-p}
\begin{equation}
\lim_{p \to \infty} \mathbb{E}\left[Z_{\tau_j^p}|\mathcal{F}_{t_j}\right]= \mathbb{E}\left[Z_{\tau_j}|\mathcal{F}_{t_j}\right]\textnormal\;{\text{in}}\; \mathbb{L}^2(\Omega) \textnormal{ for } 1 \leq j \leq N
\end{equation}
\end{proposition}

\begin{proof}
We proceed by induction. \\
For $j=N$, the proposition is true since $\tau_N^p = \tau_N = T$. Assume that the result holds for $j + 1$, Let us prove that it still holds for j:

\begin{align*}
\mathbb{E}\left[Z_{\tau_j^p } - Z_{\tau_j}|\mathcal{F}_{j_j}\right] &= Z_{t_j} \left( \mathbb{1}_{\left\{Z_{t_j} \geq \mathcal{T}_p^j(X_{t_j})\right\}} - \mathbb{1}_{\left\{ Z_{t_j} \geq \mathbb{E}\left[ Z_{\tau_{j+1}}|\mathcal{F}_{t_j}\right]\right\}}\right)\\
& + \mathbb{E}\left[Z_{\tau_{j+1}^p} \mathbb{1}_{\left\{Z_{t_j} < \mathcal{T}_p^j(X_{t_j})\right\}} - Z_{\tau_{j+1}} \mathbb{1}_{\left\{ Z_{t_j}  < \mathbb{E}\left[ Z_{\tau_{j+1}}|\mathcal{F}_{t_j}\right]\right\}} | \cF_{t_j}\right]\\
&=  (Z_{t_j} - \mathbb{E}[Z_{\tau_{j+1}} | \mathcal{F}_{t_j}]) \left( \mathbb{1}_{\left\{Z_{t_j} \geq \mathcal{T}_p^j(X_{t_j})\right\}} - \mathbb{1}_{\left\{ Z_{t_j} \geq \mathbb{E}\left[ Z_{\tau_{j+1}}|\mathcal{F}_{t_j}\right]\right\}}\right)\\
&+ \mathbb{E}[Z_{\tau_{j+1}}^p - Z_{\tau_{j+1}} | \mathcal{F}_{t_j}] \mathbb{1}_{\left\{Z_{t_j} < \mathcal{T}_p^j(X_{t_j})\right\}}\\
&= A_j^p + \mathbb{E}[Z_{\tau_{j+1}}^p - Z_{\tau_{j+1}} | \mathcal{F}_{t_j}] \mathbb{1}_{\left\{Z_{t_j} < \mathcal{T}_p^j(X_{t_j})\right\}}
\end{align*}

Where $A_j^p$ is defined by

\begin{equation*}
A_j^p = (Z_{t_j} - \mathbb{E}[Z_{\tau_{j+1}} | \mathcal{F}_{t_j}]) \left( \mathbb{1}_{\left\{Z_{j_k} \geq \mathcal{T}_p^j(X_{t_j})\right\}} - \mathbb{1}_{\left\{ Z_{t_j} \geq \mathbb{E}\left[ Z_{\tau_{j+1}}|\mathcal{F}_{t_j}\right]\right\}}\right)
\end{equation*}

On one hand, since the conditional expectation is an orthogonal projection, we have
\begin{equation*}
    \mathbb{E}\left[ \abs*{\mathbb{E}\left[ Z_{\tau_{j+1}^p} - Z_{\tau_{j+1}} | \cF_{t_j}\right]}^2\right] \leq \mathbb{E}\left[ \abs*{\mathbb{E}\left[ Z_{\tau_{j+1}^p} - Z_{\tau_{j+1}} | \cF_{t_{j+1}}\right]}^2\right]
\end{equation*}
and using the induction assumption $\mathbb{E}\left[Z_{\tau_{j+1}^p} - Z_{\tau_{j+1}} | \mathcal{F}_{t_{j+1}}\right] \to 0$ in $\mathbb{L}^2(\Omega)$ when $p \to \infty$. On the other hand,

\begin{align*}
\abs*{A_j^p} &= \abs*{Z_{t_j} - \mathbb{E}[Z_{\tau_{j+1}} | \mathcal{F}_{t_j}]} \abs*{ \mathbb{1}_{\left\{Z_{t_j} \geq \mathcal{T}_p^j(X_{t_j})\right\}} - \mathbb{1}_{\left\{ Z_{t_j} \geq \mathbb{E}\left[ Z_{\tau_{j+1}}|\mathcal{F}_{t_j}\right]\right\}}}\\
&\leq \abs*{Z_{t_j} - \mathbb{E}[Z_{\tau_{j+1}} | \mathcal{F}_{t_j}]}\;
\abs*{ \mathbb{1}_{\left\{\mathbb{E}\left[ Z_{\tau_{j+1}}|\mathcal{F}_{t_j}\right] >Z_{t_j} \geq \mathcal{T}_p^j(X_{t_j})\right\}} - \mathbb{1}_{\left\{\mathcal{T}_p^j(X_{t_j}) > Z_{t_j} \geq \mathbb{E}\left[ Z_{\tau_{j+1}}|\mathcal{F}_{t_j}\right]\right\}}}\\
&\leq \abs*{Z_{t_j} - \mathbb{E}[Z_{\tau_{j+1}} | \mathcal{F}_{t_j}]}\;
\abs*{\mathbb{1}_{\left\{ \abs*{Z_{t_j} - \mathbb{E}\left[ Z_{\tau_{j+1}}|\mathcal{F}_{t_j}\right]} \leq \abs*{\mathcal{T}_p^j(X_{t_j}) - \mathbb{E}\left[ Z_{\tau_{j+1}}|\mathcal{F}_{t_j}\right]}\right\}}}\\
&\leq \abs*{\mathcal{T}_p^j(X_{t_j}) - \mathbb{E}\left[ Z_{\tau_{j+1}}|\mathcal{F}_{t_j}\right]}\\
&\leq \abs*{\mathcal{T}_p^j(X_{t_j}) - \mathbb{E}\left[Z_{\tau_{j+1}^p}| \mathcal{F}_{t_j}\right]} + \abs*{\mathbb{E}\left[Z_{\tau_{j+1}^p}| \mathcal{F}_{t_j}\right] -\mathbb{E}\left[Z_{\tau_{j+1}}| \mathcal{F}_{t_j}\right]}.
\end{align*}

Using the induction assumption, the second term goes to zero in $\mathbb{L}^2(\Omega)$ when $p \to \infty$. Let $\left( \left[a_{i-1}(p), a_i(p) \right)\right)_{1 \leq i \leq 2^p}$ be the partition generated by $\mathcal{T}_p^j$. We define
\begin{equation*}
    \mathcal{\bar{T}}_p^j(X_{t_j}) = \sum_{i=1}^{2^p} \mathbb{E}\left[ Z_{\tau_{j+1}}| X_{t_j} \in \left[a_{i-1}(p), a_i(p) \right)\right] \mathbb{1}_{\{X_{t_j} \in \left[a_{i-1}(p), a_i(p) \right)\}}.
\end{equation*}
Note that $\mathcal{\bar{T}}_p^j$ uses the partition given by $\mathcal{T}_p^j(X_{t_j})$ but the coefficients $\alpha_i(p)$ are given by the conditional expectations of $Z_{\tau_{j+1}}$ w.r.t $X_{t_j}$ and not those of $Z_{\tau_{j+1}^p}$. Using Proposition~\ref{prop:optim-tree-weights}, we have the following inequality stating that $\mathcal{\bar{T}}_p^j(X_{t_j})$ is sub-optimal compared to $\mathcal{T}_p^j(X_{t_j})$
\begin{align*}
    & \mathbb{E}\left[ \abs*{\mathcal{T}_p^j(X_{t_j}) - \mathbb{E}\left[ Z_{\tau_{j+1}^p}|\mathcal{F}_{t_j}\right]}^2\right] \\
    & \leq \mathbb{E}\left[ \abs*{\mathcal{\bar{T}}_p^j(X_{t_j}) - \mathbb{E}\left[ Z_{\tau_{j+1}^p}|\mathcal{F}_{t_j}\right]}^2\right]\\
    &\leq 2 \mathbb{E}\left[ \abs*{\mathcal{\bar{T}}_p^j(X_{t_j}) - \mathbb{E}\left[ Z_{\tau_{j+1}}|\mathcal{F}_{t_j}\right]}^2\right] + 2 \mathbb{E}\left[ \abs*{\mathbb{E}\left[ Z_{\tau_{j+1}}|\mathcal{F}_{t_j}\right] - \mathbb{E}\left[ Z_{\tau_{j+1}^p}|\mathcal{F}_{t_j}\right]}^2\right].
\end{align*}
The second term goes to 0 using the induction assumption. As for the first term, note that the partition obtained with $\mathcal{T}_p^j$ verifies the conditions of Lemma \ref{lemma4.1}. Then, using the same arguments as in the proof of Theorem \ref{theorem4.2}, we can show that the first term also goes to $0$. 
\end{proof}
\subsection{Convergence of the Monte Carlo approximation}
In this section, the depth $p$ of the trees is fixed and we study the convergence with respect to the number of samples $M$. We assume that, for all dates $j$, the trees $\cT^j_p(X_{t_j})$ and $\hat{\cT}^{j,M}_p(X_{t_j})$ are built using the same splitting strategy (ie the same splitting directions and threshold strategies).

\subsubsection{Convergence of optimisation problems}

We present three major results to study the convergence of stochastic optimization problems related to regression trees.

The first result is a uniform strong law of large numbers, see \citep[Chap. 2, Lemma.~A1]{sharpio}, which can be seen as a particular case of the strong law of large numbers in Banach spaces \citep[Corollary 7.10, page 189]{ledoux}.
\begin{lemma}
    \label{lem:uslln}
    Let $(\xi_i)_{i \geq 1}$ be a sequence of i.i.d $\mathbb{R}^n$-valued random vectors and $h: \mathbb{R}^d \times \mathbb{R}^n \to \mathbb{R}$ be a measurable function. Assume that 
    \begin{itemize}
        \item[(i)] For all $\bar \theta \in \R^d$, the function $\theta \in \mathbb{R}^d \mapsto h(\theta, \xi_1)$ is continuous at $\bar \theta$ a.s.,
        \item[(ii)] $\forall C > 0, \mathbb{E}\left[ \sup_{\abs*{\theta} \leq C} \abs*{h(\theta, \xi_1)}\right] < \infty$.
    \end{itemize}
    Then, a.s $\theta \in \mathbb{R}^d \mapsto \frac{1}{n} \sum_{i=1}^n h(\theta, \xi_i)$ converges locally uniformly to the continuous function $\theta \in \mathbb{R}^d \mapsto \mathbb{E}\left[ h(\theta, \xi_1)\right]$, i.e
    \begin{equation*}
        \forall\, C>0, \; \lim_{n \to \infty} \sup_{\abs*{\theta} \leq C}\abs*{\frac{1}{n} \sum_{i=1}^n h(\theta, \xi_i) - \mathbb{E}\left[ h(\theta, \xi_1)\right]} = 0 \textnormal{ a.s.}
    \end{equation*}
\end{lemma}
This lemma is a slight improvement of \citep[Chap. 2, Lemma.~A1]{sharpio}, which was formulated under the assumption that the function $\theta \in \mathbb{R}^d \mapsto h(\theta, \xi_1)$ is almost surely continuous. However, looking closely at their proof, it turns that it sufficient to assume $(i)$ for the conclusion to hold. Condition $(i)$ allows the $\P-$null sets on which the continuity at $\bar \theta$ does not hold to depend on $\bar \theta$.

Consider a sequence of real valued functions $(f_n)_n$ defined on a compact set $K \subset \mathbb{R}^d$ such that there exists a sequence of $(x_n)_n$ satisfying
\begin{equation*}
    f_n(x_n) = \inf_{x \in K} f_n(x).
\end{equation*}
From \citep[Chap.~2, Theorem A1]{sharpio}, we can derive the following lemma
\begin{lemma}
\label{lem:cv-stochoptim}
Assume that the sequence $(f_n)_n$ converges uniformly on $K$ to a continuous function $f$. Let $v^* = \inf_{x \in K} f(x)$ and $\cS^* = \{x \in K \: : \: f(x) = v^*\}$. Then, $f_n(x_n) \to \inf_{x \in K} f(x)$ and $d(x_n, \cS^*) \to 0$.
\end{lemma}

Now, we focus on a canonical minimization problem appearing at each node of the regression tree. Let $[\underline{a}, \bar{a}]$ be a rectangle in $[0, 1]$. For some index $1 \le \delta \le d$, and a real number $a$ in $]\underline{a}^\delta, \bar{a}^\delta[$, we define the two new rectangles
\begin{equation}
    \cR_{\underline{a}, \bar{a}}^{\delta,l}(a) = \{x \in [\underline{a}, \bar{a}] \;:\; x^\delta \le a \}; \quad \cR_{\underline{a}, \bar{a}}^{\delta,r}(a) = \{x^\delta \in [\underline{a}, \bar{a}] \;:\; x > a \}.
\end{equation}
We consider the cost functions
\begin{align}
    \label{eq:mM}
    &m_M: y_l \in \R, y_r \in \R, x \in [\underline{a}^\delta, \bar{a}^\delta] \longmapsto \frac{1}{M} \sum_{i=1}^M \left(Y_i - y_l \ind{X_i \in \cR_{\underline{a}, \bar{a}}^{\delta,l}(x)} - y_r \ind{X_i \in \cR_{\underline{a}, \bar{a}}^{\delta,r}(x)}\right)^2 \\
    \label{eq:m}
    &m: y_l \in \R, y_r \in \R, x \in [\underline{a}^\delta, \bar{a}^\delta]  \longmapsto \E\left[ \left(Y - y_l \ind{X \in \cR_{\underline{a}, \bar{a}}^{\delta,l}(x)} - y_r \ind{X \in \cR_{\underline{a}, \bar{a}}^{\delta,r}(x)}\right)^2\right].
\end{align}
Let $(\hat{y}_l^M, \hat{y}_r^M, \hat{x}^M)$ be the solution to 
\begin{align}
    \label{eq:optim_mM}
    \inf_{y_l, y_r, x \in [\underline{a}^\delta + \varepsilon_0, \bar{a}^\delta - \varepsilon_0] } m_M(y_l, y_r, x)
\end{align}
with the smallest third component for some arbitrary small $\varepsilon_0 >0$. It is clear that~\eqref{eq:optim_mM} may not have a unique solution. Choosing the solution with the minimal third component is a standard way to get a unique minimizer (see for instance~\cite{Seijosen11}).
\begin{lemma}
    \label{lem:sargmax}
    Assume that the density $f_X$ of $X$ satisfies $f_X(x) \ge \underline{f} > 0$ for all $x$ in any compact set of $]0, 1[^d$ and that the minimization problem 
    \begin{align}
        \label{eq:optim_m}
        \inf_{y_l, y_r, x \in [\underline{a}^\delta + \varepsilon_0, \bar{a}^\delta - \varepsilon_0] } m(y_l, y_r, x)
    \end{align}
    has a unique minimizer $(y_l^*, y_r^*, x^*)$.

    Then, $(\hat{y}_l^M, \hat{y}_r^M, \hat{x}^M)$ converges almost surely to $(y_l^*, y_r^*, x^*)$ and $m_M(\hat{y}_l^M, \hat{y}_r^M, \hat{x}^M)$ converges almost surely to $m(y_l^*, y_r^*, x^*)$ when $M$ goes to infinity.
\end{lemma}
\begin{proof}
    Consider the function
    \begin{equation*}
        h(y_l, y_r, x, y, \xi) = \left(y - y_l \ind{\xi \in \cR_{\underline{a}, \bar{a}}^{\delta,l}(x)} - y_r \ind{\xi \in \cR_{\underline{a}, \bar{a}}^{\delta,r}(x)}\right)^2
    \end{equation*}
    Since $X$ has a density on $[0, 1]^d$, the function $(y_l, y_r, x) \longmapsto h(y_l, y_r, x, Y, X)$ satisfies Condition $(i)$ of Lemma~\ref{lem:uslln}. Moreover, for $C >0$

    \begin{align*}
        \E\left[ \sup_{\abs{y_l} \le C, \abs{y_r} \le C, x \in [\underline{a}^\delta + \varepsilon_0, \bar{a}^\delta - \varepsilon_0]} \abs*{h(y_l, y_r, x, Y, X)} \right] \le 2\E[Z^2] + 2C
    \end{align*}
    Hence, we deduce from Lemma~\ref{lem:uslln} that $m_M$ converges a.s. locally uniformly $m$. Note that for a fixed $x \in [\underline{a}^\delta + \varepsilon_0, \bar{a}^\delta - \varepsilon_0]$, the optimal values of $y_l$ and $y_r$ are given by 
    \begin{align*}
        y_l &= \E[Y | X \in \cR_{\underline{a}, \bar{a}}^{\delta,l}(x)] \\
        y_r &= \E[Y | X \in \cR_{\underline{a}, \bar{a}}^{\delta,r}(x)] 
    \end{align*}
    Observe that for $x \in [\underline{a}^\delta + \varepsilon_0, \bar{a}^\delta - \varepsilon_0]$
    \begin{align*}
        \abs{\E[Y | X \in \cR_{\underline{a}, \bar{a}}^{\delta,l}(x)]} &\le \frac{\sqrt{\E[Y^2]}}{\sqrt{\P(X \in \cR_{\underline{a}, \bar{a}}^{\delta,l}(x))}} \le \frac{\sqrt{\E[Y^2]}}{\sqrt{\P(X \in \cR_{\underline{a}, \bar{a}}^{\delta,l}(a^\delta - \varepsilon_0))}} 
    \end{align*}
    The uniform upper bound is uniform is finite thanks to the assumption on the density $f_X$ of $f$. Hence, as a function of $x$, $y_l$ is uniformly bounded. A similar results holds for $y_l$. Then, we deduce from Lemma~\ref{lem:cv-stochoptim}, that $m_M(\hat{y}_l^M, \hat{y}_r^M, \hat{x}^M)$ converges almost surely to $m(y_l^*, y_r^*, x^*)$ and $(\hat{y}_l^M, \hat{y}_r^M, \hat{x}^M)$ converges almost surely to $(y_l^*, y_r^*, x^*)$.
\end{proof}
To the best of our knowledge, without the uniqueness assumption on $(y_l^*, y_r^*, x^*)$, we can only prove that $(\hat{y}_l^M, \hat{y}_r^M, \hat{x}^M)$ converges in probability to $(y_l^*, y_r^*, x^*)$. See for instance~\cite[Corollary~1]{Ferger04}, which yields that for any $\eta >0$, $\P(\hat{x}^M < x^* + \eta) \to 1$. Combining this result with $d((\hat{y}_l^M, \hat{y}_r^M, \hat{x}^M), S^*) \to 0$ a.s. yields the convergence in probability of $(\hat{y}_l^M, \hat{y}_r^M, \hat{x}^M)$ to $(y_l^*, y_r^*, x^*)$.

\subsubsection{Strong law of large numbers}

In order to prove the almost sure convergence of $\hat{\cT}^M_{j,p}(X)$, we slightly modify the design of our regression trees. At every node, when computing the optimal splitting point $x^*$, we do not perform the optimization on the entire cell but we actually a margin to make sure that $y^*_l$ and $y^*_r$ are uniformly bounded w.r.t $x^*$. Consider a node at depth $p$, the optimal splitting is obtained by minimizing $m_M$ defined in~\eqref{eq:mM} where $\underline{a}$ and $\bar{a}$ were computed at depth $p-1$. As in~\eqref{eq:optim_mM}, we search for the optimal value of $x^*$ in $[\underline{a}^\delta - \varepsilon_p, \bar{a}^\delta - \varepsilon_p]$ where $\delta$ is the coordinate chosen for the splitting and $\varepsilon_p >0$ is a technical security padding decreasing to $0$. When there is no room left for this padding ($\abs{\underline{a}^\delta - \bar{a}^\delta} < 2 \varepsilon_p)$, we stop the splitting procedure.

From \citep{clement}, we have the following result
\begin{lemma}
\label{lemma4.6}
For every $j=1,\ldots, N-1$,
\begin{equation*}
    \abs*{F_j(a,Z,X) - F_j(b,Z,X)} \leq \left( \sum_{i=j}^N \abs*{Z_{t_i}}\right)\left( \sum_{i=j}^{N-1}\mathbb{1}_{\{\abs*{Z_{t_i}-\mathcal{P}_p^i(X_{t_i}, b_i)} \leq \abs*{\mathcal{P}_p^i(X_{t_i},a_i) - \mathcal{P}_p^i(X_{t_i},b_i)}\}}\right).
\end{equation*}
\end{lemma}
\begin{proposition}
\label{proposition4.7}
Assume that for all $p \in \mathbb{N}^*$, and all $1 \leq j \leq N - 1$, $\mathbb{P}(Z_{t_j} = \mathcal{P}_p^j(X_{t_j}, \theta_j^p)) = 0$. For all $1 \le j \le N-1$ and all $1 \le i \le 2^p$, the optimization problems
\begin{equation*}
    \inf_{y_l, y_r, x }\E\left[ \left(Z_{\tau^p_j} - y_l \ind{X_{t_j} \in \cR_{a^{i-1}_p, a^i_p}^{\delta,l}(x)} - y_r \ind{X_{t_j} \in \cR_{a^{i-1}_p, a^i_p}^{\delta,r}(x)}\right)^2\right]
\end{equation*}
with $x \in [a^{i-1, \delta}_p + \varepsilon_p, a^{i, \delta}_p - \varepsilon_p]$ admit a unique solution where $\delta$ is the coordinate selected by the splitting strategy.

Then, for all $j = 1, \ldots, N-1$, $\hat{\theta}_j^{p,M}$ converges a.s. to $\theta_j^p$ and $\mathcal{P}_p^{j}(X_{t_j}, \hat{\theta}_j^{p,M})$ converges a.s. to $\mathcal{P}_p^j(X_{t_j}, \theta_j^p)$ as $M \to \infty$.
\end{proposition}
\begin{proof}
We proceed by backward induction on $j$ with a nested forward induction on $p$.

{\bf \noindent Step 1: $j=N-1$}
\begin{itemize}
    \item For $p=1$, conditionally on splitting at the best point (and not at the midpoint)   the new nodes are obtained by solving 
    \begin{align*}
        {\mathcal{P}}_1^{N-1}(X_{t_{N-1}}, \hat{\theta}_{N-1}^{1,M}) = &\inf\limits_{\alpha,\beta,a} \frac{1}{M}\sum_{m=1}^M\abs*{Z_{t_N}^{(m)} - \alpha\mathbb{1}_{\{X_{t_{N-1}}^{(m)}\in \cR_{0,1}^{\delta,l}(a)\}} - \beta \mathbb{1}_{\{X_{t_{N-1}}^{(m)} \in [a, 1]\}}}^2\\
        \mathcal{P}_1^{N-1}(X_{t_{N-1}}, \theta_{N-1}^{1}) = &\inf_{\alpha,\beta,a} \mathbb{E}\left[ \abs*{Z_{t_N} - \alpha\mathbb{1}_{\{X_{t_{N-1}}\in \cR_{0,1}^{\delta,l}(a)\}} - \beta \mathbb{1}_{\{X_{t_{N-1}} \in [a, 1]\}}}^2 \right]
    \end{align*}
    where $a \in [\varepsilon_1, 1 - \varepsilon_1]$.  We can apply Lemma~\ref{lem:sargmax} to obtain ${\mathcal{P}}_1^{N-1}(X_{t_{N-1}}, \hat{\theta}_{N-1}^{1,M})$ converges to $\mathcal{P}_1^{N-1}(X_{t_{N-1}}, \theta_{N-1}^1)$ a.s as $M \to \infty$ and so does $\hat{\theta}_{N-1}^{1,M}$ converge a.s. to $\theta_{N-1}^1$. If the split is at the midpoint, the conclusion is even easier to obtain as the infimums are only computed w.r.t $\alpha$ and $\beta$ but not $a$. The same situation will occur repeatedly and for the sake of clearness, we will only treat the case of splitting at the best point which is harder to handle.

    \item  Assume that ${\mathcal{P}}_p^{N-1}(X_{t_{N-1}}, \hat{\theta}_{N-1}^{p,M})$ converges to $\mathcal{P}_p^{N-1}(X_{t_{N-1}}, \theta_{N-1}^p)$ a.s and that $\hat{\theta}_{N-1}^{p,M}$ converges a.s. to $\theta_{N-1}^p$ as $M \to \infty$ for $p\ge 1$, we will prove it for $p+1$.  For $i\in \{1,\ldots, 2^p\}$, we consider the computation of the $i-th$ node in $\hat{\cT}_{p+1}^{j,M}$ at depth $p+1$. 
    \begin{align*}
        \hat{\nu}_{p,N-1}^M(\alpha, \beta, a) &= \frac{1}{M}\sum_{m=1}^M \abs*{Z_{t_N}^{(m)} - \alpha \mathbb{1}_{\{X_{t_{N-1}}^{(m)} \in \cR_{a^{p,M}_{i-1,N-1},a^{p,M}_{i,N-1}}^{\delta,l}(a)\}} - \beta \mathbb{1}_{\{X_{t_{N-1}}^{(m)} \in \cR_{a^{p,M}_{i-1,N-1},a^{p,M}_{i,N-1}}^{\delta,r}(a)\}}}^2.
    \end{align*}
    The parameters $\hat{\theta}_{N-1}^{p+1,M}$ are obtained by computing all the nodes.  We also introduce a modified version of $\hat{\nu}_{p,N-1}^M$ in which we use the splits computed in $\cT_p^j$. 
    \begin{align*}
        \nu^M_{p,N-1}(\alpha, \beta, a)&= \frac{1}{M}\sum_{m=1}^M \abs*{Z_{t_N}^{(m)} - \alpha \mathbb{1}_{\{X_{t_{N-1}}^{(m)} \in \cR_{a_{i-1,N-1}^p,a^{p}_{i-1,N-1}}^{\delta,l}(a)\}} - \beta \mathbb{1}_{\{X_{t_{N-1}}^{(m)} \in \cR_{a^{p}_{i-1,N-1},a^{p}_{i-1,N-1}}^{\delta,r}(a)\}}}^2.
    \end{align*}
    The random functions $\nu^M_{p, N-1}$ write as standard empirical means and we will prove they are close to $\hat{\nu}_{p,N-1}^M$ for $M$ large.

    Using Lemma~\ref{lem:uslln} along with the arguments of the proof of Lemma~\ref{lem:sargmax}, it is easy to see that the random function $\alpha,\beta, a \mapsto \nu^M_{p,N-1}(\alpha, \beta, a)$ converges a.s locally uniformly to the function $\alpha, \beta, a \mapsto \mathbb{E}\left[ \abs*{Z_{t_N} - \alpha \mathbb{1}_{\{X_{t_{N-1}} \in \cR_{a_{i-1,N-1}^p,a^p_{i,N-1}}^{\delta,l}(a)\}} - \beta \mathbb{1}_{\{X_{t_{N-1}} \in \cR_{a_{i-1,N-1}^p,a^p_{i,N-1}}^{\delta,r}(a)\}}}^2\right]$. 
    \begin{align*}
        &\sup\limits_{a \in [0,1]^d, \abs*{\alpha}\leq C, \abs*{\beta} \leq C}\abs*{\hat{\nu}^M_{p,N-1}(\alpha, \beta, a) - \nu^M_{p,N-1}(\alpha, \beta, a)}\\
        &\leq \frac{1}{M}\sum_{m=1}^M \left[\abs*{2Z_{t_N}^{(m)}} + 4C \right] C \left( \mathbb{1}_{\{X_{t_{N-1}}^{(m), \delta} \in (a^{p,M,\delta}_{i-1}, a^{p,\delta}_{i-1})\}} +  \mathbb{1}_{\{X_{t_{N-1}}^{(m), \delta } \in (a^{p,M, \delta}_{i}, a^{p,\delta}_{i})\}}\right).\\
    \end{align*}
    Using the induction assumption on $p$, $a^{p,M}_i$ (resp. $a^{p,M}_{i-1}$) converges a.s. to $a_i^p$ (resp. $a_{i-1}^p$). Let $\epsilon > 0$,
    \begin{align*}
        &\limsup_M  \abs*{\hat{\nu}^M_{p,N-1}(\alpha, \beta, a) - \nu^M_{p,N-1}(\alpha, \beta, a)}\\
        &\leq \limsup_M \frac{1}{M} \sum_{m=1}^M \left[\abs*{2Z_{t_N}^{(m)}} + 4C \right] C \left( \abs*{\mathbb{1}_{\{\abs*{X_{t_{N-1}}^{(m),\delta} - a^{p,\delta}_{i}} \leq \epsilon\}}} + \abs*{\mathbb{1}_{\{\abs*{X_{t_{N-1}}^{(m),\delta} - a^{p,\delta}_{i+1}} \leq \epsilon\}}}\right)\\
        &\leq C (4C + 2\mathbb{E}\left[\abs*{2Z_{t_N}}\right]) \left(\mathbb{P}(\abs*{X_{t_{N-1}}^\delta - a^{p,\delta}_{ia}} \leq \epsilon) + \mathbb{P}(\abs*{X_{t_{N-1}}^\delta - a^{p,\delta}_{i+1}})\leq \epsilon\right).
    \end{align*}
    Since $X_{t_{N-1}}$ has a density, $\lim_{\epsilon \to 0}\mathbb{P}(\abs*{X_{t_{N-1}}^\delta - a^{p,\delta}_{i}} \leq \epsilon) = \mathbb{P}(X_{t_{N-1}}^\delta= a_p^{i,\delta}) = 0$ and $\lim_{\epsilon \to 0}\mathbb{P}(\abs*{X_{t_{N-1}}^\delta - a^{p,\delta}_{i+1}} \leq \epsilon) = \mathbb{P}(X_{t_{N-1}}^\delta= a^{p,\delta}_{i+1}) = 0$. As a result, $\abs*{\hat{\nu}^M_{p,N-1} - \nu_{p,N-1}^M} \to 0$ a.s. locally uniformly when $M \to \infty$. Thus, the random function $\hat{\nu}_{p,N-1}^M$ converges a.s. locally uniformly to the function
    \[\alpha, \beta, a \mapsto \mathbb{E}\left[ \abs*{Z_{t_N} - \alpha \mathbb{1}_{\{X_{t_{N-1}} \in \cR_{a^p_{i-1},a^p_{i}}^{\delta,l}(a)\}} - \beta \mathbb{1}_{\{X_{t_{N-1}} \in \cR_{a^p_{i-1},a^p_{i}}^{\delta,r}(a)\}}}^2\right].\]
    Now, we apply Lemma~\ref{lem:cv-stochoptim} along with the same arguments of Lemma~\ref{lem:sargmax} to conclude that ${\mathcal{P}}_{p+1}^{N-1}(X_{t_{N-1}}, \hat{\theta}_{N-1}^{p+1,M})$ converges to $\mathcal{P}_{p+1}^{N-1}(X_{t_{N-1}}, \theta_{N-1}^{p+1})$ a.s as $M \to \infty$ and that $\hat{\theta}_{N-1}^{p+1,M}$ converges a.s. to $\theta_{N-1}^{p+1}$.
\end{itemize}

{\bf \noindent Step 2: $j <N-1$.} \\
So far, we have proved that for all $p$, ${\mathcal{P}}_p^{N-1}(X_{t_{N-1}}, \hat{\theta}_{N-1}^{p,M})$ converges to $\mathcal{P}_p^{N-1}(X_{t_{N-1}}, \theta_j^p)$ a.s and that $\hat{\theta}_{N-1}^{p,M}$ converges a.s. to $\theta_j^p$ as $M \to \infty$. Now, suppose that ${\mathcal{P}}_p^k(X_{t_k}, \hat{\theta}_k^{p,M})$ (resp. $\hat{\theta}_k^{p,M}$) converges to $\mathcal{P}_p^k(X_{t_k}, \theta_k^p)$ (resp. $\theta_k^p$) a.s as $M \to \infty$ for all $p$ and for $k=N-1, \ldots, j+1$.We should prove that these convergence results hold for $j$ and this will be done by induction on $p$. Now that we have understood that considering multidimensional random variables $X_{t_j}$ does not play any role in the proof, we will make the rest of the proof as if $X$ were having in $\R$ in order to use a little lighter notation
\begin{itemize}
    \item For $p = 1$, consider
    \begin{align*}
        \hat{\nu}_{1,j}^M( \alpha, \beta, a) &= \frac{1}{M} \sum\limits_{m=1}^M \abs*{F_{j+1}\left(\hat{\vartheta}^{1,M}, Z^{(m)}, X^{(m)}\right) - \alpha \mathbb{1}_{\{X_{t_j}^{(m)} \in [0, a)\}} - \beta \mathbb{1}_{\{X_{t_j}^{(m)} \in [a, 1]} \}}^2\\
        \nu^M_{1,j}(\alpha, \beta, a) &= \frac{1}{M}\sum\limits_{m=1}^M \abs*{F_{j+1}\left(\vartheta^1, Z^{(m)}, X^{(m)}\right) - \alpha \mathbb{1}_{\{X_{t_j}^{(m)} \in [0, a)\}} - \beta \mathbb{1}_{\{X_{t_j}^{(m)} \in [a, 1]\}} }^2.\\
    \end{align*}
    The function $\nu_{1,j}^M$ writes as the sum of i.i.d random variables. Let $C \geq 0$, using Equation~\eqref{equation8}
    \begin{align*}
        &\mathbb{E}\left[\sup\limits_{a \in [0,1]^d, \abs*{\alpha} \leq C, \abs*{\beta} \leq C} \abs*{F_{j+1}\left(\vartheta^1, Z, X\right) - \alpha \mathbb{1}_{\{X_{t_j} \in [0, a)\}} - \beta \mathbb{1}_{\{X_{t_j} \in [a, 1]} \}}^2 \right]\\
        &\leq 2\mathbb{E}\left[ \abs*{F_{j+1}\left(\vartheta^1, Z, X\right)}^2\right] + 2 \mathbb{E}\left[\sup\limits_{a \in [0,1]^d, \abs*{\alpha} \leq C, \abs*{\beta} \leq C}\abs*{ \alpha \mathbb{1}_{\{X_{t_j} \in [0, a)\}} + \beta \mathbb{1}_{\{X_{t_j} \in [a, 1]\}} }^2 \right]\\
        & \leq 2 \mathbb{E}\left[ \max_{l\geq j + 1} (Z_{t_l})^2 \right] + 2C^2 < \infty.
    \end{align*}
    Using Lemma \ref{lem:uslln}, $\alpha, \beta, a \mapsto \nu_{1,j}^M(\alpha, \beta, a)$ converges a.s locally uniformly to the function $\alpha, \beta, a \mapsto \mathbb{E}\left[ \abs*{F_{j+1}(\vartheta^1, Z, X) - \alpha \mathbb{1}_{\{X_{t_j} \in [0, a)\}} - \beta \mathbb{1}_{\{X_{t_j} \in [a, 1]\}}}^2 \right]$.
    It remains to prove that $\forall C>0 \; \sup\limits_{a \in [0, 1]^d, \abs*{\alpha} \leq C, \abs*{\beta}\leq C} \abs*{\hat{\nu}_{1,j}^M(a, \alpha, \beta) - \nu_{1,j}^M(a, \alpha, \beta)} \to 0$ a.s when $M \to \infty$. \\
    Then, using Equation \eqref{equation8} and Lemma \ref{lemma4.6}
    \begin{align*}
        &\sup\limits_{a \in [0, 1]^d, \abs*{\alpha} \leq C, \abs*{\beta}\leq C}\abs*{\hat{\nu}_{1,j}^M(a, \alpha, \beta) - \nu_{1,j}^M(a, \alpha, \beta)}\\
        & \leq \sup\limits_{a \in [0, 1]^d, \abs*{\alpha} \leq C, \abs*{\beta}\leq C}\frac{1}{M}\sum_{m=1}^M \abs*{F_{j+1}\left(\hat{\vartheta}^{1,M}, Z^{(m)}, X^{(m)}\right) - F_{j+1}\left(\vartheta^1, Z^{(m)}, X^{(m)}\right)}\\
        & \abs*{F_{j+1}\left(\hat{\vartheta}^{1,M}, Z^{(m)}, X^{(m)}\right) + F_{j+1}\left(\vartheta^1, Z^{(m)}, X^{(m)}\right) -2 \alpha \mathbb{1}_{\{X_{t_j}^{(m)} \in [0, a)\}} - 2 \beta \mathbb{1}_{\{X_{t_j}^{(m)} \in [a, 1]\}}}\\
        &\leq \sum_{m=1}^M 2 \left(\max_{l\geq j+1} \abs*{Z_{t_l}^{(m)}} + 2C\right) \abs*{F_{j+1}\left(\hat{\vartheta}^{1,M}, Z^{(m)}, X^{(m)}\right) - F_{j+1}\left(\vartheta^1, Z^{(m)}, X^{(m)}\right)}\\
        &\leq \frac{1}{M}\sum_{m=1}^M 2 \left(\max_{l\geq j+1} \abs*{Z_{t_l}^{(m)}} + 2C\right) \left( \sum_{i=j+1}^{N} \abs*{Z_{t_i}^{(m)}} \sum_{i=j+1}^{N-1} \mathbb{1}_{\{\abs*{Z_{t_i}^{(m)} - \mathcal{T}_1^i(X_{t_i}^{(m)})} \leq \abs*{\hat{\mathcal{T}}_1^{i,M}(X_{t_i}^{(m)}) - \mathcal{T}_1^i(X_{t_i}^{(m)})}\}}\right).\\
    \end{align*}
    Using the induction assumption on $j$, $\hat{\mathcal{T}}_1^{i,M}(X_{t_i}^{(m)})$ converges a.s. to $\mathcal{T}_1^i(X_{t_i}^{(m)})$ for all $N-1 \ge i \ge j+1$. Let $\epsilon<0$.
    \begin{align*}
        &\limsup_M \sup\limits_{a \in [0, 1]^d, \abs*{\alpha} \leq C, \abs*{\beta}\leq C} \abs*{\hat{\nu}_{1,j}^M(a, \alpha, \beta) - \nu_{1,j}^M(a, \alpha, \beta)} \\
        &\leq \frac{1}{M}\sum_{m=1}^M 2 \left(\max_{l\geq j+1} \abs*{Z_{t_l}^{(m)}} + 2C\right) \left( \sum_{i=j+1}^{N} \abs*{Z_{t_i}^{(m)}} \sum_{i=j+1}^{N-1} \mathbb{1}_{\{\abs*{Z_{t_i}^{(m)} - \mathcal{T}_1^i(X_{t_i}^{(m)})} \leq \epsilon\}}\right).
    \end{align*}
    Since $\mathbb{P}(Z_{t_j}^{(m)} = {\mathcal{T}}_p^j(X_{t_j}^{(m)})) = 0 $, then $\lim_{\epsilon \to 0} \mathbb{1}_{\{\abs*{Z_{t_i}^{m)} - \mathcal{T}_1^i(X_{t_i}^{(m)})} \leq \epsilon\}} = 0$ a.s and we conclude that a.s. $\abs*{\hat{\nu}_{1,j}^M(a, \alpha, \beta) - \nu_{1,j}^M(a, \alpha, \beta)}$ converges to zero uniformly. Thus, $\hat{\nu}_{1,j}^M$ converges a.s uniformly to the function $a, \alpha, \beta \mapsto \mathbb{E}\left[ \abs*{F_{j+1}(\vartheta^1, Z, X) - \alpha \mathbb{1}_{\{X_{t_j} \in [0, a)\}} - \beta \mathbb{1}_{\{X_{t_j} \in [a, 1]\}}}^2 \right]$. Then, we apply Lemma~\ref{lem:uslln} and the arguments of the proof of Lemma~\ref{lem:sargmax} to deduce that ${\mathcal{P}}_1^{j}(X_{t_{j}}, \hat{\theta}_{j}^{1,M})$ converges to $\mathcal{P}_j^{j}(X_{t_{j}}, \theta_j^j)$ a.s and that $\hat{\theta}_{j}^{j,M}$ converges a.s. to $\theta_j^p$ as $M \to \infty$.
    \item We assume that ${\mathcal{P}}_p^k(X_{t_k}, \hat{\theta}_k^{p,M})$ (resp. $\hat{\theta}_k^{p,M}$) converges to $\mathcal{P}_p^k(X_{t_k}, \theta_k^p)$ (resp. $\theta_k^p$) a.s as $M \to \infty$ for some $p$ and for $k=N-1, \ldots, j$. We will probe that it holds for $p+1$.
    
    Let $i\in \{1, \ldots, 2^p\}$ and consider  
    \begin{align*}
        \hat{\nu}_{p,j}^M(\alpha, \beta, a) &= \frac{1}{M}\sum_{m=1}^M \abs*{F_{j+1}(\hat{\vartheta}^{p,M}, Z^{(m)}, X^{(m)}) - \alpha \mathbb{1}_{\{X_{t_j}^{(m)} \in [a^{p,M}_{i - 1,j}, a)\}} - \beta \mathbb{1}_{\{X_{t_j}^{(m)} \in [a, a^{p,M}_{i, j})\}}}^2.\\
        \nu_{p,j}^M(\alpha, \beta, a)&= \frac{1}{M}\sum_{m=1}^M \abs*{ F_{j+1}(\vartheta^p, Z^{(m)}, X^{(m)})- \alpha \mathbb{1}_{\{X_{t_j}^{(m)} \in [a^p_{i-1,j}, a)\}} - \beta \mathbb{1}_{\{X_{t_j}^{(m)} \in [a, a^p_{i,j})\}}}^2
    \end{align*}
    The function $\nu_{p,j}^M$ writes as the sum of i.i.d random variables. Let $C \geq 0$, 
    \begin{align*}
        &\mathbb{E}\left[\sup\limits_{a \in [0,1]^d, \abs*{\alpha} \leq C, \abs*{\beta} \leq C}\abs*{F_{j+1}\left(\vartheta^p, Z^{(m)}, X^{(m)}\right) - \alpha \mathbb{1}_{\{X_{t_j}^{(m)} \in [a^p_{i-1,j}, a)\}} - \beta \mathbb{1}_{\{X_{t_j}^{(m)} \in [a, a^p_{i,j})\}} }^2 \right]\\
        &\leq 2\mathbb{E}\left[ \abs*{F_{j+1}\left(\vartheta^p, Z^{(m)}, X^{(m)}\right)}^2\right] + 2 C^2\\
        & \leq 2 \mathbb{E}\left[ \max_{l\geq j + 1} (Z_{t_l})^2 \right] + 2C^2 < \infty.
    \end{align*}
    We conclude using Lemma~\ref{lem:uslln} that a.s $ \nu_{p,j}^M$ converges locally uniformly to the function \[\alpha, \beta, a \mapsto \mathbb{E}\left[ \abs*{F_{j+1}(\vartheta^p, Z, X) - \alpha \mathbb{1}_{\{X_{t_j} \in [a^p_{i-1, j}, a)\}} - \beta \mathbb{1}_{\{X_{t_j} \in (a, a^p_{i,j})\}}}^2 \right].\] 
    Let $C > 0$, 
    \begin{align*}
        &\sup\limits_{a \in [0, 1]^d, \abs*{\alpha} \leq C, \abs*{\beta}\leq C} \abs*{ \hat{\nu}_{p,j}^M(\alpha, \beta, a) -  \nu_{p,j}^M(\alpha, \beta, a)}\\
        &\leq \sup\limits_{a \in [0, 1]^d, \abs*{\alpha} \leq C, \abs*{\beta}\leq C} \frac{1}{M}\sum_{m=1}^M \left(2\max_{l\geq j+1} \abs*{Z_{t_l}^{(m)}} + 4C\right) \\
        &\left[ \left( \left(\sum_{i=j+1}^{N} \abs*{Z_{t_i}^{(m)}}\right) \sum_{i=j+1}^{N-1} \mathbb{1}_{\{\abs*{Z_{t_i} - \mathcal{T}_p^i(X_{t_i}^{(m)})} \leq \abs*{\hat{\mathcal{T}}_p^{i, M}(X_{t_i}^{(m)}) - \mathcal{T}_p^i(X_{t_i}^{(m)})}\}}\right) \right.\\
        & \left.+ \abs*{ \alpha \mathbb{1}_{\{X_{t_j}^{(m)} \in [a^p_{i-1,j}, a^{p,M}_{i-1,j})\}} + \beta \mathbb{1}_{\{X_{t_j}^{(m)} \in [a^{p,M}_{i, j}, a^p_{i,j}]\}} }\right].
    \end{align*}
    Using the induction assumption on $p$, $\lim_{M \to \infty} a^{p,M}_{i-1,j} = a^p_{i-1,j}$ and $\lim_{M \to \infty} a_p^{i,M} = a_p^i$ a.s and using the induction assumption on $j$, $\lim_{M \to \infty} \hat{\mathcal{T}}_p^{i, M}(X_{t_i}^{(m)}) = \mathcal{T}_p^i(X_{t_i}^{(m)}) \; \forall i \ge j+1$. Let $\epsilon <0$, 
    \begin{align*}
        & \limsup_M \sup\limits_{a \in [0, 1]^d, \abs*{\alpha} \leq C, \abs*{\beta}\leq C} \abs*{ \hat{\nu}_{p,j}^M(\alpha, \beta, a) -  \nu_{p,j}^M(\alpha, \beta, a)}\\ 
        &\leq \limsup_M  \frac{1}{M}\sum_{m=1}^M  \left(2\max_{l\geq j+1} \abs*{Z_{t_l}^{(m)}} + 4C\right)\\
        &\left[ C \left( \mathbb{1}_{\{\abs*{X_{t_j}^{(m)} - a^p_{i-1,j}} \leq \epsilon\}} + \beta \mathbb{1}_{\{\abs*{X_{t_j}^{(m)} - a^p_{i,j}} \leq  \epsilon\}}\right) +\left( \left(\sum_{k=j+1}^{N} \abs*{Z_{t_k}^{(m)}} \right)\sum_{k=j+1}^{N-1} \mathbb{1}_{\{\abs*{Z_{t_k}^{(m)} - \mathcal{T}_p^k(X_{t_k}^{(m)})} \leq \epsilon\}}\right)\right].
    \end{align*}
Since $\lim_{\epsilon \to 0 } \mathbb{1}_{\{\abs*{Z_{t_i} - \mathcal{T}_p^i(X_{t_i}^{(m)})} \leq \epsilon\}} = \lim_{\epsilon \to 0} \abs*{ \alpha \mathbb{1}_{\{\abs*{X_{t_j}^{(m)} - a^p_{i-1,j}} \leq \epsilon\}} + \beta \mathbb{1}_{\{\abs*{X_{t_j}^{(m)} - a^p_{i,j}} \leq  \epsilon\}} } = 0$, we conclude that a.s $\abs*{ \hat{\nu}_{p,j}^M(\alpha, \beta, a) -  \nu_{p,j}^M(\alpha, \beta, a)} \to 0$ locally uniformly when $M \to \infty$, and thus the random function $\hat{\nu}_{p,j}^M$ converges a.s locally uniformly to the function $\alpha, \beta, a \mapsto \mathbb{E}\left[ \abs*{F_{j+1}(\vartheta^p, Z, X) - \alpha \mathbb{1}_{\{X_{t_j} \in [a_p^{i-1}, a)\}} - \beta \mathbb{1}_{\{X_{t_j} \in (a, a_p^i)\}}}^2 \right]$. Then we apply Lemma~\ref{lem:cv-stochoptim} along with the arguments of the proof of Lemma~\ref{lem:sargmax} to conclude that ${\mathcal{P}}_{p+1}^k(X_{t_k}, \hat{\theta}_k^{p+1,M})$ (resp. $\hat{\theta}_k^{p+1,M}$) converges to $\mathcal{P}_{p+1}^k(X_{t_k}, \theta_k^p)$ (resp. $\theta_k^{p+1}$) a.s as $M \to \infty$ for $k=N-1, \ldots, j$. This completes the induction.
\end{itemize}
\end{proof}
\begin{theorem}
\label{thm:slln}
Assume that for all $p \in \mathbb{N}^*$, and all $1 \leq j \leq N - 1$, $\mathbb{P}(Z_{t_j} = \mathcal{T}_p^j(X_{t_j})) = 0$. Then, for $\alpha=1,2$ and for every $j=1,\ldots, N$,
\begin{equation*}
    \lim\limits_{M\to\infty} \frac{1}{M}\sum\limits_{i=1}^M \left( Z_{\tau_j^{p,(m)}}^{(m)}\right)^{\alpha} = \mathbb{E}\left[(Z_{\tau_j^p})^{\alpha}\right] \textnormal{a.s}.
\end{equation*}
\end{theorem}
\begin{proof}
Note that $\mathbb{E}\left[ (Z_{\tau_j^p})^{\alpha}\right] = \mathbb{E}\left[ F_j(\vartheta^p,Z, X)^{\alpha} )\right]$ and by the strong law of large numbers \begin{equation*}
    \lim\limits_{M \to \infty} \frac{1}{M}\sum_{m=1}^M F_j\left((\vartheta^p, Z^{(m)}, X^{(m)})^{\alpha}\right) = \mathbb{E}\left[ F_j(\vartheta^p,Z, X)^{\alpha} \right] \textnormal{a.s.}
\end{equation*}
It remains to prove that
\begin{equation*}
    \Delta F_M = \frac{1}{M}\sum\limits_{m=1}^M F_j(\hat{\vartheta}^{p,M}, Z^{(m)}, X^{(m)})^{\alpha} - F_j(\vartheta^p, Z^{(m)}, X^{(m)})^{\alpha} \xrightarrow[M \to \infty]{a.s} 0.
\end{equation*}
For any $x,y \in \mathbb{R}$, and $\alpha=1,2$, $\abs*{x^{\alpha} - y^{\alpha}} \leq \abs*{x-y} \abs*{x^{\alpha - 1} + y^{\alpha - 1}}$. Using Lemma \ref{lemma4.6} and Equation \eqref{equation8}, we have
\begin{align*}
    \abs*{\Delta F_M} &\leq \frac{1}{M} \sum\limits_{m=1}^M \abs*{F_j(\hat{\vartheta}^{p,M}, Z^{(m)}, X^{(m)})^{\alpha} - F_j(\vartheta^p, Z^{(m)}, X^{(m)})^{\alpha}}\\
    & \leq 2 \frac{1}{M} \sum\limits_{m=1}^M  \sum_{i=j}^{N} \max\limits_{k \geq j} \abs*{Z_{t_k}^{(m)}}^{\alpha - 1}\abs*{Z_{t_i}^{(m)}} \sum_{i=j}^{N-1} \mathbb{1}_{\{\abs*{Z_{t_i}^{(m)} - \mathcal{T}_p^i(X_{t_i}^{(m)})} \leq \abs*{\hat{\mathcal{T}}_p^{i, M}(X_{t_i}^{(m)}) - \mathcal{T}_p^i(X_{t_i}^{(m)})}\}}.
\end{align*}
Using Proposition \ref{proposition4.7}, for all $i=j, \ldots,N-1$, $\abs*{\hat{\mathcal{T}}_p^{i,M}(X_{t_i}) - \mathcal{T}_p^i(X_{t_i})} \to 0$ a.s when $M \to \infty$. Then for any $\epsilon > 0$,
\begin{align*}
    &\limsup_M \abs*{\Delta F_M}\\
    & \leq 2 \limsup_M \frac{1}{M}\sum\limits_{m=1}^M\sum\limits_{i=j}^N \max\limits_{k \geq j} \abs*{Z_{t_k}^{(m)}}^{\alpha - 1}\abs*{Z_{t_i}^{(m)}} \sum_{i=j}^{N-1} \mathbb{1}_{\{\abs*{Z_{t_i}^{(m)} - \mathcal{T}_p^i(X_{t_i}^{(m)})} \leq \epsilon \}}\\
    &\leq 2 \mathbb{E}\left[ \sum\limits_{i=j}^N \max\limits_{k \geq j} \abs*{Z_{t_k}}^{\alpha - 1}\abs*{Z_{t_i}} \sum_{i=j}^{N-1} \mathbb{1}_{\{\abs*{Z_{t_i} - \mathcal{T}_p^i(X_{t_i})} \leq \epsilon \}}\right].
\end{align*}
We conclude that $\limsup_M \abs*{\Delta F_M} = 0$ by letting $\epsilon$ go to 0 which ends the proof. 
\end{proof}
Proving the global convergence of our algorithm amounts to studying the difference
\begin{align*}
    \frac{1}{M}\sum_{i=1}^M  Z_{\tau_j^{p,(m)}}^{(m)} - \mathbb{E}\left[(Z_{\tau_j})\right] = \left(\frac{1}{M}\sum_{i=1}^M Z_{\tau_j^{p,(m)}}^{(m)} - \mathbb{E}\left[(Z_{\tau_j^p})\right]\right) + \left(\mathbb{E}\left[Z_{\tau_j^p}\right] - \mathbb{E}\left[(Z_{\tau_j^p})\right]\right)
\end{align*}
By Proposition~\ref{prop:pricecv-p}, the $\mathbb{E}\left[(Z_{\tau_j^p})\right] - \mathbb{E}\left[Z_{\tau_j^p}\right]$ goes to zero when $p \to \infty$ independently of $M$. Even if Theorem~\ref{thm:slln} proves the a.s. convergence of the first term for any fixed $p$ when $M \to \infty$, its behaviour seems less clear when both $M$ and $p$ go to infinity. Note that the same difficult already occurred in~\cite{clement} for the more standard least square approach. The main difficulty stems for the inability to compute the variance of the limiting distribution appearing in the central limit theorem governing the convergence of $\frac{1}{M}\sum_{i=1}^M \left( Z_{\tau_j^{p,(m)}}^{(m)}\right)$ even in the least square framework. In particular, we do seem to control the behaviour of the limiting variance with respect to $p$.

\section{Random forests}
Regression trees are barely used as standalone estimators but are often aggregated to obtain a random forest. Considering our training sample $\cD_M = \{ (X_i, Y_i)_{1 \le i \le M}\}$ of size $M$. Let $\Theta$ be a random variable independent of $\cD_M$ and used to resample the training set without replacement. Note that the size of the resampled set is usually smaller that $M$. Let $\Theta_1, \dots, \Theta_B$ $B$ iid samples of $\Theta$. For every $k=1,\dots,B$, $\hat \cT_p^M(X, \Theta_k)$ denotes the regression tree of depth $p$ computed on the resampled set obtained from $\cD_M$ by using $\Theta_k$. Then, the random forest estimator is defined by 
$$\mathcal{H}_{B,p}(X) = \sum_{k=1}^B \frac{1}{B}\hat\cT^M_{p,\Theta_k}(X).$$
We also introduce $\mathcal{H}_p(X) = \mathbb{E}_{\Theta}\left[ \hat\cT^M_{p,\Theta}(X)\right] = \lim_{B \to \infty} \mathcal{H}_{B,p}(X)$.

\begin{theorem}
\label{Theorem5.2}
\begin{equation*}
    \lim_{B \to \infty} \mathbb{E}\left[ \abs*{Y - \mathcal{H}_{B,p}(X)}^2\right] = \mathbb{E}\left[ \abs*{Y - \mathcal{H}_p(X)}^2\right] 
\end{equation*}
See Theorem 11.1 in \citep{Breiman}.
\end{theorem}
\begin{theorem}
\label{Theorem5.3}
\begin{equation*}
    \mathbb{E}\left[ \abs*{Y - \mathcal{H}_p(X)}^2\right] \leq \bar{\rho} \mathbb{E}_{\Theta}\left[ \mathbb{E}\left[ \abs*{Y - \hat \cT^M_{p,\Theta}(X)}^2\right]\right]
\end{equation*}
where $\bar{\rho}$ is the weighted correlation between the residuals $Y - \hat \cT^M_{p,\Theta}(X)$ and $Y - \hat \cT^M_{p,\Theta^{\prime}}(X)$ and $\Theta$ and $\Theta^{\prime}$ are independent. See Theorem 11.2 in \citep{Breiman}
\end{theorem}
Theorem \ref{Theorem5.3} says that to have a good generalization error in the random forest, one should have small generalization errors in the basis trees, and the basis trees should not be highly correlated.
\section{Numerical results}
\subsection{Description}
This section studies the price of some Bermudan options using regression trees or random forests to approximate the conditional expectations. We compare the results to some reference prices and those given by the standard Longstaff Schwarz method with regression on polynomial functions, which is a basic LSM algorithms. More sophisticated representations can be used to mildly temper the curse of dimensionality. We use the Scikit-Learn library in Python, \citep{scikit-learn}. For regression trees, this library offers two methods of splitting: ``best'' to choose the best split, meaning that the split threshold is the one that minimizes the MSE and the direction for splitting is the one that gives the lowest MSE among all directions; ``random'' to choose the best random split, meaning that the split threshold is the one that minimizes the MSE and the direction for splitting is chosen randomly. For the following tests, we will use the latter method, which is just slightly different from what we presented in Section \ref{section2} in the way that no mid-point cuts will be considered. We also use the feature \texttt{min\_samples\_leaf} which allows us to set a minimum number of samples in each node. This will allow us to avoid over-fitting. For random forests, we will use the bootstrapping method (\texttt{Bootstrap=True}), meaning that for each tree in the forest, we will use a sub-sample drawn randomly and with replacement from the training data. We will also use the feature \texttt{max\_samples} which allows having a specific number of data points or a percentage of the training data attributed to each tree. Having the trees trained on different data as much as possible allows us to have a low correlation between the trees which, using Theorem \ref{Theorem5.3}, should make the random forest more robust.\\
Following the work of \citep{LongstaffSchwarz}, we only use the in-the-money paths to learn the continuations values, which significantly improves the numerical computations. All the prices that we show are obtained after resimulation, meaning that the paths used in the estimation of the conditional expectations are not the same ones used by the Monte Carlo which means that the prices we show are unbiased. 

For small dimensional problems, our algorithm takes approximately the same computation time as the LSM approach for a comparable accuracy. To measure the computational time of our algorithm, we take into account both the training and running times. However, note that the training step over weights the prediction step by a great deal. Once the trees are trained, evaluating them is straightforward, which makes the resimulation price almost free. 
The computational time of the random forest approach linearly depends on the number of regression trees in the forest.

\subsection{Black and Scholes}
Consider the Black and Scholes model
\begin{equation*}
    \begin{cases}
        dS_t^i &= rS_t^i dt + \sigma_i S_t^idB_t^i,\\
        d<B^i,B^j>_t &= \rho_{ij}dt.
    \end{cases}
\end{equation*}
where $\sigma_i$ is the volatility of the underlying $S^i$, assumed to be deterministic, $r$ is the interest rate, assumed constant, and $\rho_{ij}$, represents the correlation between the underlyings $S^i$ and $S^j$, assumed constant.  
\subsubsection{One-dimensional put option}
We consider the Bermudan put option with payoff $(K -S_{\tau})^+$ with maturity $T= 1$ year, $K=110$, $S_0= 100$, $\sigma=0.25$, exercisable at $N=10$ different dates. We consider $r=0.1$. We have a reference price for this option of 11.987 computed by a convolution method in \citep{lord}. The LSM algorithm converges to the correct price with only a polynomial of degree 3. Figure \ref{figure1}, shows the price of the option when we use regression trees with a random split strategy (continuous line) or a best split strategy (dotted line) to estimate the conditional expectations. With the random strategy, the best price we get is 11.89. The case \texttt{min\_samples\_leaf}=1 and \texttt{max\_depth}=20 gives a price of 10.5, which is far from the reference price. This result is due to over-fitting. In fact, for this case, the number of degrees of freedom is too big. The tree fits the training data too well, but it cannot generalize when confronted with new data. For the best split strategy, we obtain a slightly better price of 11.94. However, depending on the tree parameters, the price fluctuates, and we can see that the best split strategy is not necessarily better than the random split strategy. Thus, for the following, we will keep using the random split strategy. Random forests with basis trees of maximum depth 5 and minimum 100 samples in each leaf converge to the correct price with only ten trees. 
\begin{figure}[H]
    \centering
    \includegraphics{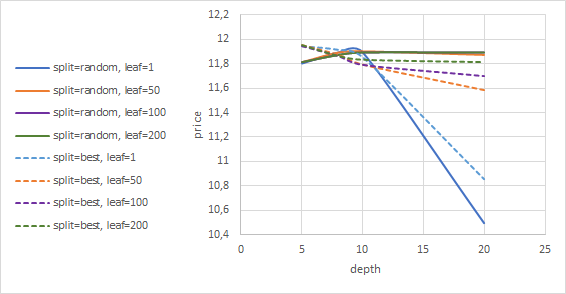}
    \caption{one dimensional put with regression trees, true price=11.987}
    \label{figure1}
\end{figure}
\subsubsection{Call option on the maximum of two assets}
We consider a call option on the maximum of 2 assets with payoff  $(\max(S_{\tau}^1, S_{\tau}^2) - K)^+$, we use the same set of parameters as in \citep{glasserman}, for which we have reference prices of 13.90, 8.08 and 21.34 for $S_0^i=100, 90$ and $110$ respectively. The LSM algorithm using a polynomial of degree 5 converges to a price of 13.90, 8.06, 21.34 for the cases $K=100, 90, 110$ respectively. This is a small dimensional problem, so the convergence of the LSM is expected. With regressions trees we have slightly less satisfying results as shown in Figure \ref{figure5}. We can still see the case of over-fitting when giving the regression trees too many degrees of freedom. Aggregating the regression trees into random forests immediately improves the results as shown in Figure \ref{figure6}. Note that the lower the percentage of data in each basis tree, the better the results. This confirms the results of Theorem \ref{Theorem5.3} .
\begin{figure}[H]
    \centering
    \includegraphics[scale=0.7]{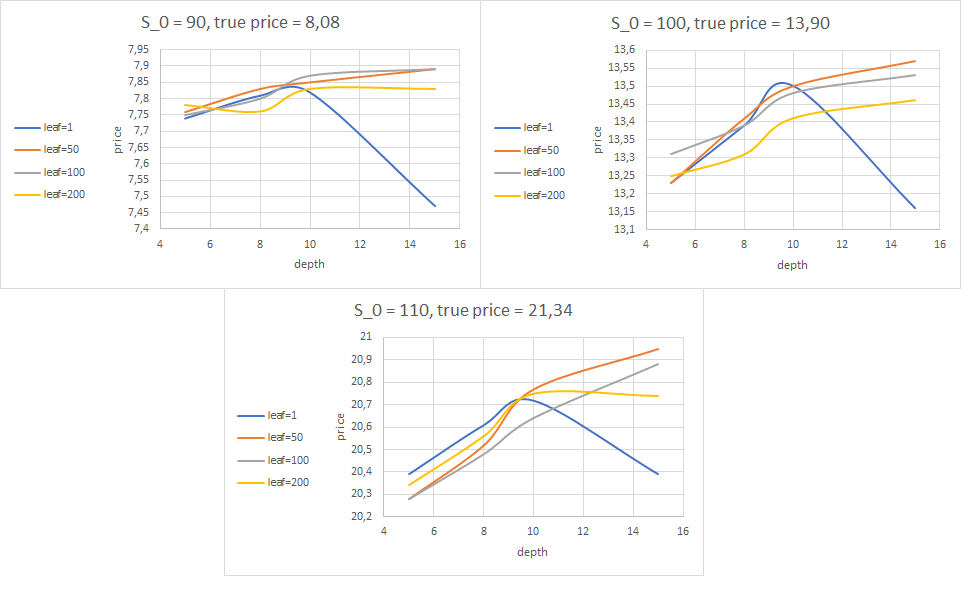}
    \caption{Call on the maximum of two assets with regression trees, $K=100, T=3$ years, $\sigma^i =0.2, r=0.05, \rho_{ij}=0, \delta_i=0.1, N=9, M=100,000$}
    \label{figure5}
\end{figure}

\begin{figure}[H]
    \centering
    \includegraphics[scale=0.7]{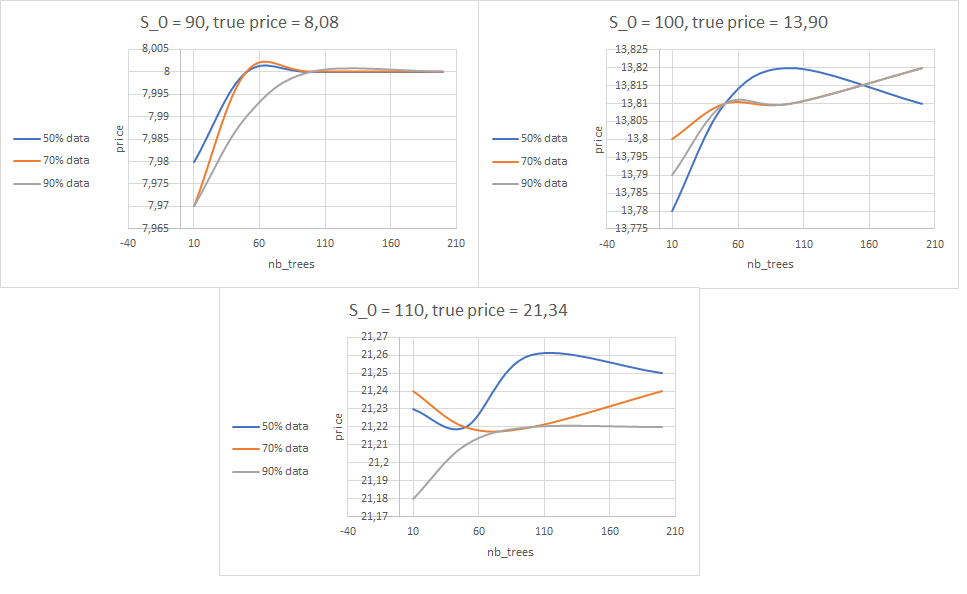}
    \caption{Call on the maximum of two assets with random forests, $K=100, T=3$ years, $\sigma^i =0.2, r=0.05, \rho_{ij}=0, \delta_i=0.1, N=9, M=100,000$}
    \label{figure6}
\end{figure}

\subsubsection{Geometric basket option}
We consider a Bermudan Put option on a geometric basket of $d$ underlying with payoff $$\left(K - \left(\prod\limits_{i=1}^d S_{\tau}^i\right)^{\frac{1}{d}}\right)^+.$$ We test the following option for $d=2,10,40$ for which we have reference prices from \citep{cox} using the CRR tree method. With the LSM algorithm, we converge to the correct price 4.57 for the case $d=2$, using only a polynomial of degree 3. For the case $d=10$, we can at most use a polynomial of degree 3 due to the curse of dimensionality. With this parametrization, we obtain a price of 2.90 for a true price of 2.92. For the case $d=40$, we cannot  go further than a polynomial of degree 1, which yields a price of 2.48 for a reference price of 2.52. Figure \ref{figure2} shows the results obtained with regression trees. For the case $d=2$, the best price we get is 4.47 and, as expected, the LSM algorithm has a better performance. This is also the case for the cases $d=10$ and $d=40$ where the best prices we obtain are 2.84 and 2.46 respectively. Notice that even though these are high dimensional cases, the trees converge with only a depth of 5 or 8. For $d=10$, our algorithm with \texttt{depth} equal $8$ and \texttt{min\_samples\_leaf} equal $200$ runs in $3$ seconds compared to the $11$ seconds needed by the LSM with polynomial degree $3$. The running times of the two approaches are pretty comparable. For $d=40$, our algorithm takes the same computational time as the LSM with degree $1$ polynomials. We also notice the importance of the parameter \texttt{min\_samples\_leaf}. In fact, letting the trees grow without managing this parameter (case leaf1) leads to a problem of over-fitting. The results get better when we use random forests as shown in Figure \ref{figure3}. For these random forests, we used basis trees of \texttt{max\_depth=8} and \texttt{min\_samples\_leaf=100}. Notice for the case $d=2$, the curve where only 50\% of the data is used gives much better results as in this case the basis trees are the less correlated. For the cases $d=10$ and $d=40$, the best choice is not necessarily to use 50\% of the data in each tree. As these are larger dimensions, having the trees trained on a small percentage of the training data maybe not enough. One may consider extending the size of the training data itself. Furthermore, we notice that once the percentage of data to use in each tree is chosen, the price of the option converges as the number of trees in the forest grows. However, note the computational time of the random forest method linearly depends on the number of trees inside the forest. Although the regression tree approach runs within roughly the same computational time as the LSM for a similar accuracy, the random forests approach may take much longer as the number of trees increases.
\begin{figure}[H]
    \centering
    \includegraphics[scale=0.7]{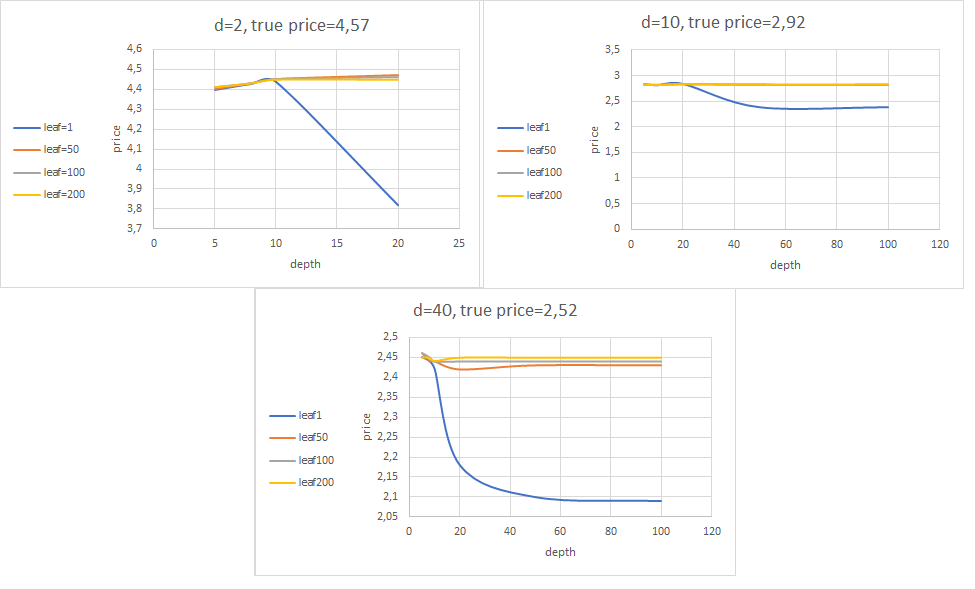}
    \caption{Geometric put option with regression trees}
    \label{figure2}
\end{figure}
\begin{figure}[H]
    \centering
    \includegraphics[scale=0.7]{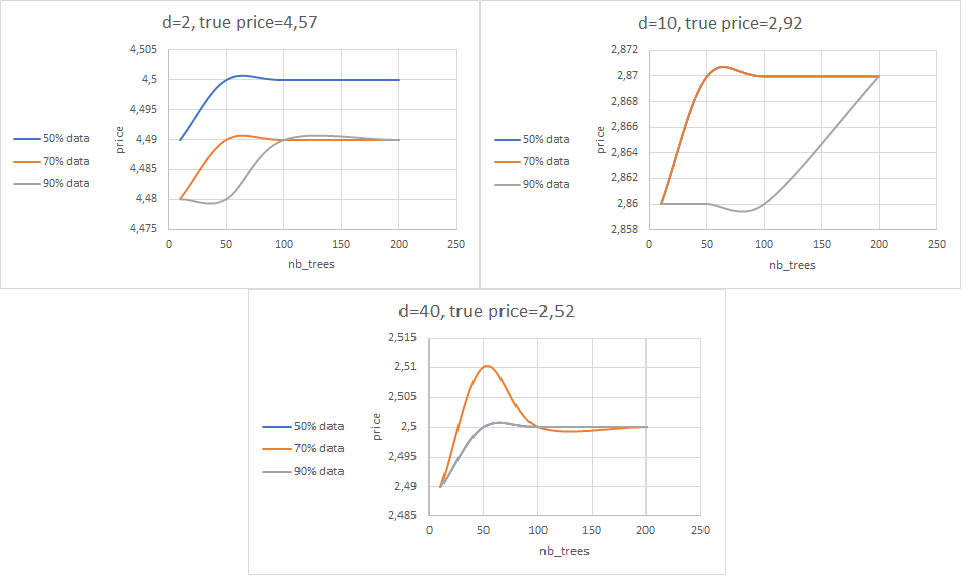}
    \caption{Geometric put option with random forests}
    \label{figure3}
\end{figure}

\subsubsection{A put basket option}
We consider a put option on the basket of $d=40$ asset with payoff $\left( K - \sum_{i=1}^d \omega_i S_T^i\right)^+$. We test this payoff for $d=40$ for which we have a reference price from \citep{goudenege} between 2.15 and 2.22 using the following set of parameters: $T=1, S_i=100, K=100, r=0.05, \sigma_i=0.2, \rho_{ij}=0.2, \omega_i = \frac{1}{d}$ and $N=10$. With a polynomial of degree 1, we obtain a price of 2.15 using the LSM algorithm. The results obtained with regression trees are shown in Figure \ref{figure7}.

\begin{figure}[H]
    \centering
    \includegraphics{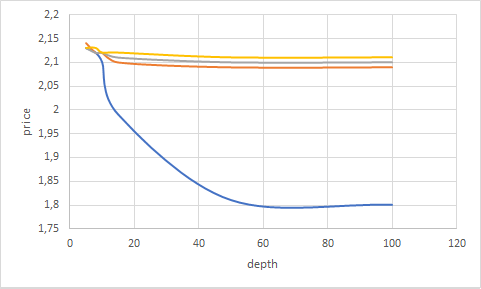}
    \caption{Put on a basket of 40 asset with regression trees}
    \label{figure7}
\end{figure}
Even though this example is high dimensional, we do not need a lot of parameters to estimate the conditional expectations (the trees converge for very small depths). This will not be the case for the next example which is very non linear. The aggregation into random forests leads to a price of 2.16 using only 50 trees.   

\subsubsection{A call on the max of 50 asset}
We consider a call option on the maximum of $d=50$ asset with payoff $\left(\max\limits_{1 \leq i \leq d}(S_T^i) - K \right)^+$ with the following characteristics: $K=100, T = 3$ years, $ S_0^i = 100, \sigma_i = 0.2, \delta_i = 0.1, \rho_{ij} = 0 \;\forall i,j, r=0.05, N=9, M=100 000$. \citep{becker} report $[69.56, 69.95]$ as the 95\% confidence interval for the option price. With the LSM algorithm we find a price of 67.88 with a polynomial of degree 1. This is a difficult example for which bigger trees are needed to approach the conditional expectations. At maturity, the payoff depends only on one direction (corresponding to the best performance), if the cuts in the tree never consider that direction, the estimation will not be correct. As a result, we consider a number of cuts big enough to ensure that each direction is taken into consideration. We allow the depth to grow while monitoring the  \texttt{min\_samples\_leaf} in order to have a significant number of samples in each leaf. Table \ref{table1} shows the results obtained with regression trees. As the best price we obtain is given by depth=100 and \texttt{min\_samples\_leaf} = 100, we use this set of parameters for the random forest part. Table \ref{table2} shows the results that we obtain with this method. 

\begin{table}[H]
    \centering
    \begin{tabular}{c|c|c}
         depth &  \texttt{min\_samples\_leaf} & price \\
         \hline
         50 & 50 & 66,89\\
         50 & 100 & 66.88\\
         100 & 50 & 67.13\\
         100 & 100 & 67.31\\
         200 & 50 & 67.16\\
         200 & 100 & 67.28\\
    \end{tabular}
    \caption{A call option on the maximum of 50 asset with regression trees}
    \label{table1}
\end{table}

\begin{table}[H]
    \centering
    \begin{tabular}{c|c|c}
        \texttt{nb\_trees} & \texttt{max\_samples} & price\\
         \hline
         10 & 50\% & 68,32\\
         10 & 70\% & 68,32\\
         10 & 90\% & 68,29 
    \end{tabular}
    \caption{A call option on the maximum of 50 asset with random forests}
    \label{table2}
\end{table}
Using only regression trees is not enough to have acceptable results. However, as soon as we aggregate the regressor into random forests, we obtain very satisfying results and with just 10 trees we converge to a good price. We can also notice in this example that using uncorrelated trees leads to better results (see the case \texttt{max\_samples}= 50\% or 70\% against the case \texttt{max\_samples} = 90\%). 

\subsection{A put in the Heston model}
We consider the Heston model defined by
\begin{align*}
    dS_t &= S_t(r_tdt + \sqrt{\sigma_t}(\rho dW_t^1 + \sqrt{1 - \rho^2}dW_t^2))\\
    d\sigma_t &= \kappa (\theta - \sigma_t)dt + \xi \sqrt{\sigma_t}dW_t^1
\end{align*}
and we consider a put option with payoff $\left(K - S_T\right)^+$. we have no reference price for this option, so we will just compare the results of regression trees and random forests to the LSM method. We use the following set of parameters: $K=100, S_0=100, T = 1, \sigma_0=0.01, \xi=0.2, \kappa=2, \rho=-0.3, r=0.1, N=10$ and $M=100,000$. The LSM method yields a price of 1.70. Figures \ref{figure8} and \ref{figure9} show the results obtained with regression trees and random forests. Both methods converge to the same price of LSM. We notice for this example the occurrence of the over-fitting phenomenon for regression trees with \texttt{max\_depth=15} and \texttt{min\_sample\_leaf=1}. We also have the same behavior for random forests in function of the percentage of data given to each basis tree. 
\begin{figure}[H]
    \centering
    \includegraphics[scale=0.7]{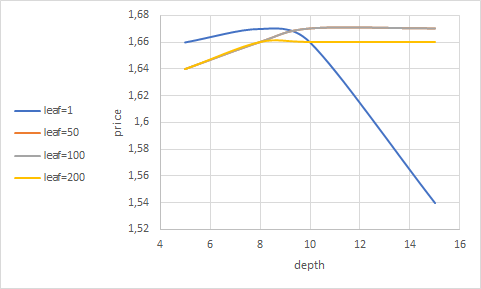}
    \caption{A put option in the Heston model with regression trees}
    \label{figure8}
\end{figure}
\begin{figure}[H]
    \centering
    \includegraphics[scale=0.7]{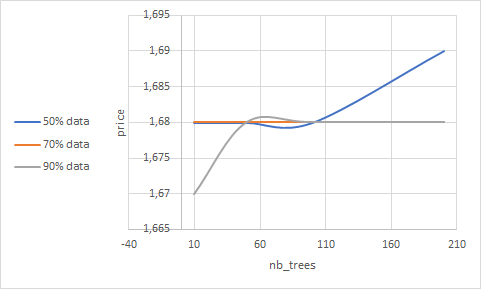}
    \caption{A put option in the Heston model with random forests}
    \label{figure9}
\end{figure}
\section{Conclusion}
Pricing Bermudan options comes down to solving a dynamic programming equation where the main trouble comes from the computation of the conditional expectations representing the conditional expectations. We have explored the usage of regression trees and random forests for the computations of these quantities. We have proved in two steps the convergence of the algorithm when regression trees are used: first, the convergence of the conditional expectations; Then, the convergence of the Monte Carlo approximation. This problem was particularly hard to solve given that the regression trees do not solve a global optimization problem as does the functional regression used in the LSM algorithm. We have shown through numerical experiments that we obtain good prices for some classical examples using regression trees. The aggregation of regression trees into random forests yields even better results. We came to the conclusion that for small dimensional problems, a simpler algorithm like the LSM is efficient enough. However, for high dimensional problems, it is interesting to consider using random forests. Instead of using all the features of the problem, the basis trees in the forest only use a subset of the features which can help combat the problem of the curse of dimensionality. 

\bibliographystyle{plainnat} 
\bibliography{references}

\end{document}